\documentclass[peerreview,a4paper]{IEEEtran}
\usepackage[utf8]{inputenc}

\renewcommand{\baselinestretch}{0.99}

\usepackage[T1]{fontenc}
\usepackage{amssymb}
\usepackage{amsmath}
\usepackage{amsthm}
\usepackage{manfnt}

\usepackage{soul}
\usepackage{bbm}
\usepackage{graphicx}
\usepackage{subfiles}
\usepackage[dvipsnames]{xcolor}
\usepackage{comment}
\usepackage{dsfont}
\usepackage{epsfig,cite}
\usepackage{float}
\usepackage{tikz}
\usepackage[caption=false]{subfig}
\usepackage{nicefrac}
\usetikzlibrary{positioning}
\usetikzlibrary{arrows}
\tikzstyle{block}=[draw opacity=0.7,line width=1.4cm]
\usetikzlibrary{positioning,shapes,shadows,arrows,scopes,decorations}
\usetikzlibrary{matrix,chains,shapes.geometric,shapes.arrows,automata,chains,er,fit}
\tikzstyle{comment}=[rectangle, draw=black, fill=red!50!black, rounded corners, drop shadow,
anchor=west, text=white]
\usetikzlibrary{shapes.callouts}
\usetikzlibrary{mindmap}
\usetikzlibrary{calc,through,backgrounds}
\usepackage{ifthen}
\usepackage{hyperref}
\newboolean{OMMITISIT}
\setboolean{OMMITISIT}{true}

\newtheorem{theorem}{Theorem}
\newtheorem{corollary}{Corollary}
\newtheorem{lemma}{Lemma}

\newtheorem{definition}{Definition}
\newtheorem{remark}{Remark}

\newtheorem{example}{Example}

\newcommand{\Var}{\mathsf{Var}}
\newcommand{\E}{\mathbb{E}}
\newcommand{\sign}{sign}
\newcommand{\eps}{\varepsilon}

\newcommand{\Gfunc}{d}

\newcommand{\fM}{f}

\usepackage[square,numbers]{natbib}

\title{Maximal correlation under cardinality constraints}

\usepackage{authblk}
\author[1]{Dror Drach, Tomer Berg, Or Ordentlich, Ofer Shayevitz}
\graphicspath{ {./} }

\def\biblio{\bibliographystyle{IEEEtran}\bibliography{references}}

\begin{document}

\def\biblio{}

\maketitle

\begin{abstract}
 In this paper, we define and analyze the quantized maximal correlation, an extension of the notion of maximal correlation restricted to functions taking values in sets of bounded cardinality. We derive an upper bound on the quantized maximal correlation by showing that the correlation between any quantized functions of $X$ and $Y$ is related to the  MMSE distortion in quantization of a particular linear combination of random variables. Following this, we leverage rate-distortion techniques and anti-concentration inequalities to further bound this MMSE, which results in explicit bounds on the quantized maximal correlation. Unlike the quantized maximal correlation itself, which does not generally tensorize, our bounds on the mean squared error do tensorize, resulting in a dimension-free upper bound on the quantized maximal correlation for product distributions. Our results also lead to improved bounds on the isoperimetric constants of reversible Markov chains and product chains, strengthening classical results such as those by Alon and Milman.
\end{abstract}

% \textcolor{red}{
% To Do
%     \item \ofer{May want to add a reference and a short discussion (maybe as a footnote) of the paper "Ricci curvature of finite Markov chains via convexity of the entropy", https://arxiv.org/pdf/1612.00514.pdf, which gives another lower bound on the isoperimetric constant as a function of a square root of the spectral gap and the smallest transition rate (for reversible chains with a nonnegative entropic Ricci curvature). We may want to cite this paper as well https://arxiv.org/pdf/1111.2687.pdf .  I am not sure if our example chain has nonnegative Ricci, but even if it does their bound is most likely still worse than ours, seems it's good only around $\lambda_1\approx 0$.}
% \end{itemize}
% }

%%%%%%%%%%%%%%%%%%%%%%%%%%%%%%%%%%%%%
%          A PROPER SKELETON        %
%%%%%%%%%%%%%%%%%%%%%%%%%%%%%%%%%%%%%
\section{Introduction}   
Let $(X,Y)~\sim P_{XY}$ be two correlated random variables, and let $f(X)$ and $g(Y)$ denote the outputs of an $M$ and $N$ level quantizers operating on $X$ and $Y$ respectively. We are interested in characterizing the maximal correlation that can be attained between any quantizers outputs $f(X)$ and $g(Y)$. To motivate our work, let us consider the case of $M=N=2$, which is directly related to the fundamental problem of analyzing the probability of disagreement between Boolean functions of dependent random variables~\cite{Witsenhausen_Main,borell1985geometric,am85,o2014analysis}. Here one is often interested in lower-bounding $\Pr(f(X)  \neq  g(Y))$ in terms of the correlation between $f(X)$ and $g(Y)$. It is easy to see that for any $f:\mathcal{X}\to \{0,1\}$ and $g:\mathcal{Y}\to \{0,1\}$ such that $p=\E [f(X)]$ and $q=\E [g(Y)]$, it holds that 
\begin{align}
    \Pr(f(X)  \neq  g(Y))& \geq 2\sqrt{p(1-p)q(1-q)}(1-\rho(f,g)), \label{eq:AgreementProbWitsBound}
\end{align}
where $\rho(f,g)$ is the Pearson correlation between $f(X)$ and $g(Y)$. To obtain a uniform bound, it is natural to define the \textit{binary maximal correlation} $\rho_b(X;Y) $ between $X$ and $Y$, as the maximum Pearson correlation achievable by any pair of zero-mean, unit-variance binary valued functions of $X$ and $Y$. That is,
\begin{align}
    &\rho_b(X;Y) \triangleq \sup_{f,g} \E[ f(X) g(Y)]\label{eq:binbin}
\end{align}
where the maximization is over all $f, g$ taking exactly two distinct values, subject to $\mathbb{E}[f(X)] =\mathbb{E}[g(Y)] = 0$ and $\mathbb{E}[f^2(X)]= \mathbb{E}[g^2(Y)] =1$. By construction, $\rho(f(X),g(Y)) \le \rho_b(X;Y)$ for any Boolean $f,g$, so substituting $\rho_b$ into eq.~\eqref{eq:AgreementProbWitsBound} yields the universal lower bound 
\begin{align}
    \Pr(f(X)  \neq  g(Y))& \geq 2\sqrt{p(1-p)q(1-q)}(1-\rho_b(f,g)). \label{eq:rhoProbWitsBound}
\end{align}%This reduction of a functional inequality to a property of the random variables $(X,Y)$ itself has proved useful in a range of settings, for example, in combinatorics and theoretical computer science, inequalities of this type underpin analysis of Boolean functions and influence of variables~\cite{o2014analysis}.

Despite its significance, binary maximal correlation $\rho_b(X;Y)$ is notoriously challenging to compute and analyze in general due to the finite cardinality constraint since, as the alphabet sizes of $X$ or $Y$ grow, searching over all Boolean functions becomes intractable. Moreover, a counterexample given by Bradley~\cite{bradley2016cousin} shows that binary maximal correlation does not tensorize under product distributions. Hence computing its value for high-dimensional functions over i.i.d. pairs, a task that is often interesting, is generally infeasible. These difficulties motivate the introduction of the so called \textit{Hirschfeld–Gebelein–Rényi maximal correlation}, denoted $\rho_m(X;Y)$. The maximal correlation removes the binary restriction, allowing $f$ and $g$ to be any zero-mean, unit-variance functions, potentially taking a continuum of values. First introduced by Hirschfeld and Gebelein~\cite{hirschfeld1935connection} and later formalized by Rényi~\cite{renyi1959measures}, $\rho_m(X;Y)$ captures the largest possible correlation between any functions of $X$ and $Y$. A remarkable fact, proved via functional analysis techniques by Lancaster~\cite{LancasterComplete1958}, is that $\rho_m(X;Y)$ is equal to the second singular value in a canonical spectral decomposition of the joint distribution $P_{XY}$ (see Section~\ref{subsec:spectralDecomp}). For the case of discrete $X$ and $Y$, the maximal correlation is given by the second singular value in the singular value decomposition of the so called \textsf{DTM} matrix, $$\textsf{DTM}=\text{diag}(P_X^{-1/2})\cdot P_{XY}\cdot\text{diag}(P_Y^{-1/2}),$$ with the first singular value being 1, corresponding to the trivial constant functions. Consequently, $\rho_m(X;Y)$ is much easier to compute than $\rho_b(X;Y)$ in the scalar case. Even more important is the tensorization property proved in~\cite{Witsenhausen_Main}, which states that for i.i.d. pairs, $\rho_m(X^n;Y^n) = \rho_m(X;Y)$. These properties make the maximal correlation a reasonable substitute for $\rho_b$ for product distributions $\prod_{i=1}^nP_{X_iY_i}$, leading to a tensorized form of the bound in~\eqref{eq:AgreementProbWitsBound},
\begin{align}
    \Pr(f(X^n) \neq g(Y^n))& \geq 2\sqrt{p(1-p)q(1-q)}(1-\rho_m(X,Y)).\label{eq:Witsen3}
\end{align}
This is however a weakening of the disagreement inequality for Boolean functions. Indeed, while $\rho_m(X,Y)=1$ implies $\rho_b(X,Y)=1$ and a zero disagreement probability for Boolean functions, for general $\rho_m(X,Y)<1$ the bound can be loose.

Beyond its direct definition, binary maximal correlation is closely connected to several classical problems. One notable connection is to isoperimetric inequalities on discrete spaces. Indeed, if one considers the special case $f=g$, then $\Pr(f(X)\neq f(Y))$ in Witsenhausen’s setting becomes the probability that $f$ takes different values on two correlated copies of $X$. In the context of Markov chains, this quantity is precisely related to the edge expansion (\textit{Cheeger constant}) of the chain~\cite{cheeger1969lower,PeresLevin2017markov}. For example, if $W$ is a reversible Markov transition kernel on state space $\mathcal{X}$ with stationary distribution $\mu$, the isoperimetric (Cheeger) constant $h(W)$ can be written as\footnote{In the literature, the Cheeger constant is usually defined as $h(W)/2$.} 
\begin{align}
  h(W) = \inf_{f:\mathcal{X}\to \{0,1\}} \frac{\Pr(f(X)\neq f(Y))}{\min\{\Pr(f(X)=0),\Pr(f(X)=1)\}}  
\end{align}
with $(X,Y)\sim \mu\times W$. Witsenhausen’s inequality~\eqref{eq:rhoProbWitsBound}  then implies a lower bound on $h(W)$ in terms of the binary maximal correlation of $(X,Y)$. Specifically, noting that $\min\{p,1-p\} \leq 2p(1-p)$, we have
\begin{align}
    h(W) &\geq  1-\rho_b(X;Y), \label{eq:h_bin_bound}
\end{align} 
which in fact can be shown to be tight up to a multiplicative factor of $2$ in some cases, e.g., when $W$ is a lazy reversible kernel and $\mu$ is uniform~\cite{PeresLevin2017markov}.
A powerful application of eq.~\eqref{eq:Witsen3} is the celebrated isoperimetric inequality $h(W^n) \geq n^{-1}(1-\rho_m(X;Y))$, where $W^n$ is the $n$-fold Cartesian product of $W$, which is a form of the famous inequality of Alon and Milman for product spaces~\cite{am85}.

In this work, we derive new bounds on the binary maximal correlation which we extend to arbitrary quantization levels. 
To that end, we define the $(M,N)$-quantized maximal correlation $\rho_{M,N}(X;Y)$ as the maximum of $\mathbb{E}[f(X)g(Y)]$ over functions $f,g$ taking at most $M$ and $N$ distinct values respectively.  We show that the inner product between any two functions $f(X)$ and $g(Y)$ can be upper bounded by their inner product on the maximal correlation subspace (see definition of $\mathcal{H}_{f_2}^n$ in~\ref{subsec:productspace}) and the orthogonal subspace, which in turn implies a bound that depends on the second and third singular values in the canonical spectral decomposition of the joint distribution, and on the projection of any $M$-valued function of $X^n$ (resp. $N$-valued function of $Y^n$) onto the space spanned by the maximal correlation attaining functions. This projection is shown to be related to the minimum mean square error in $M$-level quantization of a unit norm linear combination of random variables. Intuitively, this implies that if $\rho_m(X;Y)$ is significantly larger than $\rho_{M,N}(X;Y)$, then any attempt to quantize the optimal correlating functions into $M$ and $N$ levels must incur a large MSE loss. 

This result allows us to derive upper bounds on quantized maximal correlation using bounds on MSE distortion of source coding / quantization problems. We formalize this through two approaches: the first one appeals to classical rate distortion techniques where the input is a linear combination of i.i.d. random variables, the rate is $R=\log M$ bits and the distortion is quadratic. The second approach leverages results due to Petrov~\cite{petrov2012sums} and Esséen~\cite{esseen1968concentration} on the concentration functions of sums of random variables. In particular, building on the Kolmogorov–Rogozin inequality and its refinements~\cite{kolmogorov1956two,rogozin1961increase}, we show that large anti-concentration of linear combination of random variables implies a large distortion in $M$-level quantization of said combination.
Finally, we obtain improved isoperimetric inequalities for Markov Chains and product Markov chains. In particular, for a reversible Markov chain with transition kernel $W$, we improve the Cheeger constant $h(W)$ lower bounds obtained in the classic results of Alon and Milman~\cite{am85} for graphs and the bounds of Houdré and Tetali~\cite{houdre2004isoperimetric} for Markov chains. %Likewise, the improved isoperimetric bounds have direct implications for Markov chain mixing times via Cheeger’s inequality; since the Cheeger constant provides a handle on the spectral gap and thus mixing rate, our results suggest faster mixing (or smaller bottlenecks) in certain chains. 
%The information-theoretic viewpoint we employ for maximal correlation ties in with the study of contraction coefficients in information theory. 
%Recent works by Makur et al.~\cite{makur2019information} have highlighted how such contraction coefficients can be used to characterize limits of non-interactive simulation, functional inequalities, and more. Our findings add to this narrative by providing concrete bounds and relationships that could be useful in those settings. 
%Finally, we note that maximal correlation has seen renewed interest in fields like graphical models and machine learning as a measure of nonlinear association that can detect dependencies, for instance, maximal correlation has been used as a criterion for feature selection and fairness in machine learning models (see, e.g.,~\cite{huang2019universal,lee2022maximal,gavel2022maximum,huang2024universal}). We hope that the improved understanding and bounds for quantized maximal correlation developed here may translate into more effective algorithms for those applications (for example, controlling the loss of dependency information when data is discretized or quantized for privacy or fairness).

The paper is organized as follows. In Section \ref{sec:prelim}, we provide background on maximal correlation, including formal definitions of the  quantized maximal correlation and a review of the canonical spectral decomposition. Section \ref{sec:upper} develops the core upper bound on quantized maximal correlation, and in Section \ref{sec:gaussian} we use this result to obtain an improved upper bound on the binary maximal correlation of jointly Gaussian random variables. In Section \ref{sec:ratedist} and Section \ref{sec:anticont} we derive lower bounds on MSE distortion in quantization of normalized linear combinations of random variables using information theoretic rate distortion tools and anti-concentration techniques respectively. In Section \ref{sec:examples} we evaluate the bounds obtained in previous sections for different distributions and show that different approaches (i.e., rate distortion, anti concentration) give tighter bounds in different scenarios. Finally, in Section \ref{sec:apps} we discuss applications of our results, and particularly obtain an improved lower bound on the isoperimetric (Cheeger) constant of reversible Markov chains.

\section{Preliminaries}\label{sec:prelim}

\subsection{Notations and definitions}\label{subsec:notations}
Throughout the paper we consider (either discrete, continuous or mixed) jointly distributed i.i.d. random variables $(X^n,Y^n) \sim P_{XY}^{\otimes n}$ with marginals $P_X$ and $P_Y$, respectively. We let $f,g$ represent (usually finite-valued) functions that operate on random variables. We let the inner product $\langle f(X),g(Y)\rangle$ denote the cross-correlation between $f(X)$ and $g(Y)$, that is, $\langle f(X),g(Y)\rangle=\E[f(X)g(Y)]$. Define $\rho_m(X;Y)$ as the maximal correlation between $X$ and $Y$, where $\rho_b(X;Y)$ is the maximal correlation between any binary function of $X$ and $Y$. Let $H(X)=-\sum_{i=1}^mp_i\log p_i$ be the Shannon entropy of a discrete r.v. $X$ with pmf $p$ supported on $[m]$, and $h(Y)=-\int_{\mathbb{R}}f(y)\log f(y)dy$ be the differential entropy of a continuous r.v. $Y$ with pdf $f(y)$ supported on $\mathbb{R}$. The standard Gaussian probability density function and cumulative distribution function are denoted by $\varphi(x) = \frac{1}{\sqrt{2\pi}}e^{-\frac{x^2}{2}}$ and $\Phi(x) = \int_{-\infty}^x \varphi(x)$ respectively. Finally, we let $\mathcal{F}_M^n(P_X)$ denote the family of all zero mean, unit norm, $M$-level functions of $X^n$, that is, we say that $f\in\mathcal{F}_M^n(P_X)$ if $\fM:\mathcal{X}^n\to \{c_1,\ldots,c_M\}$ for some real values $\{c_1,\ldots,c_M\}$, and also $\E (\fM(X^n)) = 0,\E (\fM^2(X^n)) = 1$. For brevity, we let $\mathcal{F}_M^1(P_X)=\mathcal{F}_M(P_X)$.

\begin{definition}\label{def:quant_maxcorr}
 The \textit{$(M,N)$-quantized maximal correlation} is defined as 
 \begin{align*}
    \rho_{M,N}(X;Y) & \triangleq \sup_{f\in\mathcal{F}_M^n(P_X),g\in\mathcal{F}_N^n(P_Y)} \E [f(X)g(Y)].
\end{align*} 
When $M=N$ we simply write $\rho_{M}(X;Y)$. Note that $\rho_2(X;Y) = \rho_b(X;Y)$, and $ \rho_{\infty}(X;Y)=\rho_m(X;Y)$. We also write $\rho_{\infty,N}(X;Y)$ for the \textit{one-sided quantized maximal correlation}, i.e., where there is no restriction on the cardinality of $X$.
\end{definition}
%%%%%%%%%%%%%%%%%%%%%%%%%%%%%%%%%%%%%%
%    Spectral theory  a.k.a DTM      %
%%%%%%%%%%%%%%%%%%%%%%%%%%%%%%%%%%%%%%
\subsection{Spectral decomposition of joint distributions}\label{subsec:spectralDecomp}
It was shown in~\cite{LancasterComplete1958} (see also \cite{makur2019information}) that under some mild regularity conditions there are (possibly countably infinite) orthonormal sets of singular functions $\{f_i:\mathcal{X}\to \mathbb{R}\}_{i\geq 1}$ and $\{g_j:\mathcal{Y}\to \mathbb{R}\}_{j\geq 1}$ spanning $L^2(P_X)$ and $L^2(P_Y)$ respectively, satisfying\footnote{These orthonormal sets exist for any joint distribution with finite $\chi^2$-information. In particular, this always holds for discrete alphabets, in which case the expansion corresponds to the standard singular value decomposition of the DTM matrix whose entries are $P_{XY}(x,y)/\sqrt{P_X(x)P_Y(y)}$~\cite{Witsenhausen_Main}.} 
\begin{align*}
    \langle f_i(X),f_j(X)\rangle&=\langle g_i(Y),g_j(Y)\rangle=\mathds{1}(i=j)\\ \langle f_i(X),g_j(Y)\rangle&=\sigma_i \mathds{1}(i=j)
\end{align*}
for nonnegative singular values $1=\sigma_1 \geq \sigma_2 \geq \cdots $, where $f_1 = g_1 = 1$ and $\sigma_2 = \rho_m(X;Y)$. For finite alphabets, $\sigma_i = 0$ for any $i > \min \{|\mathcal{X}|,|\mathcal{Y}|\}$, and $f_{i} =g_{j} = 0$ for any $i>|\mathcal{X}|$ and $j > |\mathcal{Y}|$. We refer to the set $\{f_i, g_i, \sigma_i\}_{i\geq 1}$ as the \textit{canonical system} of $P_{XY}$.

\section{A general upper Bound on Quantized Maximal Correlation }\label{sec:upper}

\subsection{Duality}

For a closed subspace $H \subseteq L^2(\pi)$ and function $f \in L^2(\pi)$, we write $f_H$ for the orthogonal projection of $f$ onto $H$, i.e., the unique element of $H$ satisfying $\langle f- f_H, h\rangle = 0$ for all $h \in H$. This can be construed as the  element of $H$ closest to $f$. We write $H^c$ for the orthogonal complement of $H$ and $f_{H^c} = f - f_H$.

As convention, all inner products and norms are taken in $L^2(\pi)$, i.e., $\langle f, g \rangle = \E_\pi \left[f(X) g(X) \right]$, and $\mathcal{F}_M=\mathcal{F}_M(\pi)$ denotes the family of zero mean, unit norm, $M$-level functions under $\pi$.

\begin{definition}[Optimal $M$-level MMSE of a subspace]\label{def:MMSEopt}
Let $H\subseteq L^2(\pi)$ be a nonzero closed subspace. The optimal MMSE in $M$-level quantization of $H$ is
\begin{align}
    D_M(H)\triangleq \inf_{\substack{\fM:\mathcal{X}\to \{c_1,\ldots,c_M\} \\ h\in H:\|h\|_2=1}}\E\left(f(X)-h(X)\right)^2. 
\end{align}
\end{definition}

\begin{definition}[Maximal projection onto a subspace]\label{def:maxproj}
The maximal projection of an $M$-valued function onto $H$ is\footnote{Note that, since $||f_H||_2 = \sup_{\substack{h\in H\\ ||h||_2=1}} \langle f, h \rangle$, we have
$$\phi_M(H) = \sup_{f\in\mathcal{F}_M} \sup_{\substack{h\in H\\ ||h||_2=1}} \langle f, h \rangle.$$} 
\begin{align}
  \phi_M(H) \triangleq \sup_{f\in\mathcal{F}_M} ||f_H||_2. 
\end{align}
\end{definition}

\begin{theorem}\label{thm:ProjQuantDuality}[Projection-quantization duality]
For any nonzero closed subspace $H\subseteq L^2(\pi)$ orthogonal to the constant functions,
$$(\phi_M(H))^2 + D_M(H) = 1.$$
\end{theorem}

\begin{remark}
    Both $\phi_M(H)$ and $D_M(H)$ are determined by the pair $(\pi, H)$ alone. In particular, when a joint law enters the picture in the next subsection, it will do so only through the choice of a subspace $H$ for each of its marginals.
\end{remark}

\begin{proof}
Any (non-constant) $M$-valued function $Q$ admits the canonical decomposition $Q = \mu + \sigma\tilde f$ with $\mu \triangleq \mathbb{E}[Q]$, $\sigma \triangleq \|Q - \mu\|_2$, $\tilde f \triangleq (Q-\mu)/\sigma \in \mathcal{F}_M$. Then, for any unit-norm $h \in H$,
$$\mathbb{E}[(Q - h)^2] = \mu^2 + \sigma^2 - 2\sigma\langle\tilde f, h\rangle + 1 = \mu^2 + (\sigma - \langle\tilde f, h\rangle)^2 + 1 - \langle\tilde f, h\rangle^2,$$
which over $(\mu, \sigma) \in \mathbb{R}\times\mathbb{R}_{\geq 0}$ attains $1 - \langle\tilde f, h\rangle^2$ at $\mu = 0$, $\sigma = \langle\tilde f, h\rangle$ (WLOG $\geq 0$ by sign-flip closure of $H$). Hence
$$D_M(H) = \inf_{Q, h}\mathbb{E}\left[(Q - h)^2\right] = \inf_{\tilde f \in \mathcal{F}_M,\, h}\bigl(1 - \langle\tilde f, h\rangle^2\bigr) = 1 - \sup_{\tilde f \in \mathcal{F}_M,\, h}\langle\tilde f, h\rangle^2.$$
It remains to evaluate this joint maximum. Writing $\tilde f = \tilde f_H + \tilde f_{H^c}$ and recalling that $h \perp H^c$, Cauchy–Schwarz gives
$$\langle\tilde f, h\rangle = \langle\tilde f_H, h\rangle \leq \|\tilde f_H\|,$$
with equality at $h = \tilde f_H/\|\tilde f_H\|$ when well-defined (both sides vanish otherwise). For each $\tilde f$ the optimal $h$ depends only on $\tilde f$ itself, so the joint maximum collapses to a single-variable one:
$$D_M(H) = 1 - \sup_{\tilde f \in \mathcal{F}_M,\, h}\langle\tilde f, h\rangle^2 = 1 - \sup_{\tilde f \in \mathcal{F}_M}\|\tilde f_H\|^2 = 1 - \left(\phi_M(H)\right)^2.$$
\end{proof}

\subsection{Spectral decomposition}\label{sec:spect}

Let $P_{XY}$ be a joint law with canonical system $\{f_i,g_i,\sigma_i\}_{i\geq 1}$ (Section~\ref{subsec:spectralDecomp}). Recall that $\{f_i\}_{i\geq 1}$ and $\{g_i\}_{i\geq 1}$ are orthonormal bases of $L^2(P_X)$ and $L^2(P_Y)$ respectively,
satisfying $\langle f_i(X),g_j(Y)\rangle=\sigma_i \mathds{1}(i=j)$
with $f_1=g_1=\boldsymbol{1}$ and  $1=\sigma_1 \geq \sigma_2 \geq \cdots $. For any $I \subseteq \{2, 3, \ldots \}$, we define the $I$-subsystem of $P_{XY}$ as the subset $\{f_i,g_i,\sigma_i\}_{i\in I}$ characterized by the subspaces 
\begin{align}
   H_{P_X}(I) &= \overline{\mathrm{span}}\{f_i\}_{i\in I} \\H_{P_Y}(I) &= \overline{\mathrm{span}}\{g_i\}_{i\in I}
\end{align}
and $\sigma(I)=\sup_{i\in I} \sigma_i$, with $\sup \emptyset = 0$.
Note that, expanding $f'\in H_{P_X}(I)$ and $g'\in H_{P_Y}(I)$ in the subsystem and using the bi-orthogonality of the canonical system,
\begin{align}
    \langle f', g'\rangle = \sum_{i\in I}\langle f', f_i\rangle\langle g', g_i\rangle\,\sigma_i \leq \sigma(I)\,||f'||_2\,||g'||_2,\label{eq:subsystem_corr}
\end{align}
and the bound is approached by the pairs $(f_i, g_i)$, $i\in I$, with $\sigma_i \to \sigma(I)$. We accordingly refer to $\sigma(I)$ as the \textit{maximal correlation of the subsystem}.
Similarly, the bi-orthogonality of the canonical system implies bi-orthogonality between matched subspaces of disjoint subsystems, i.e., for $f'\in H_{P_X}(I)$ and $g'\in H_{P_Y}(I')$ with $I \cap I' = \emptyset$, we have $\langle f', g'\rangle = 0$. Let $I^*$ denote the complement of $I$ within $\mathbb{N}\setminus 1$, and for brevity, let $f_I\triangleq f_{H_{P_X}(I)}$ and $g_I\triangleq g_{H_{P_Y}(I)}$ be the orthogonal projection of $f$ onto $H_{P_X}(I)$ and the orthogonal projection of $g$ onto $H_{P_Y}(I)$ respectively. In the following, we apply Cauchy Schwarz inequality to derive an upper bound on $\langle f,g\rangle $.

\begin{lemma}[Subsystem split]\label{Lemma:BasicFullPaper}
Let $I \subseteq \{2, 3, \ldots \}$. For any zero-mean, unit-norm functions $f\in L^2(P_X),g\in L^2 (P_Y)$,
\begin{align}
 \langle f,g\rangle  \leq  \sigma(I) \sqrt{ ||f_{I}||_2^2 \cdot||g_{I}||_2^2} + \sigma(I^*) \sqrt{(1-||f_{I}||_2^2)(1-||g_{I}||_2^2)}.\label{eq:SpectralLemmaTwoSubsystems}
\end{align}
\end{lemma}

\begin{proof}
First note that, since $I \cap I^* = \emptyset$,
\begin{align}
    \langle f_I,g_{I^*}\rangle=\langle f_{I^*},g_I\rangle=0,
\end{align}
which implies
\begin{align}
    \langle f, g \rangle  = \langle f_{I},g_{I}\rangle + \langle f_{I^*},g_{I^*}\rangle .\label{LemmaFullProofA}
\end{align}
Applying eq.~\eqref{eq:subsystem_corr} to each term, we have $\langle f_I, g_I \rangle \leq \sigma(I)\, ||f_I||_2\, ||g_I||_2$ and $\langle f_{I^*}, g_{I^*} \rangle \leq \sigma(I^*)\, ||f_{I^*}||_2\, ||g_{I^*}||_2$. Finally, since $\{f_i\}_{i\geq 1}$ and $\{g_i\}_{i\geq 1}$ are orthonormal bases and $f, g$ are zero-mean with unit norms, we have $||f_{I^*}||_2^2 = 1 - ||f_{I}||_2^2$ and $||g_{I^*}||_2^2 = 1 - ||g_{I}||_2^2$. Substituting into~\eqref{LemmaFullProofA} yields~\eqref{eq:SpectralLemmaTwoSubsystems}, as desired.
\end{proof}

Motivated by the upper bound of Lemma~\ref{Lemma:BasicFullPaper}, we define, for parameters $s, t \geq 0$:
\begin{align}
    r(x,y) \triangleq  s \cdot x y + t \cdot \sqrt{(1-x^2)(1-y^2)}, \label{eq:d_ConvexWeightSingularValues}
\end{align}
and
\begin{align}
    d(a,b)=\sup_{\substack{0\leq x\leq a \\ 0\leq y\leq b}}r(x,y).\label{eq:dfunc}
\end{align}

\begin{definition}[Dominant subsystem]\label{def:subsys_dominance}
    An $I$-subsystem is called \textbf{dominant} if $I$ contains an index attaining $\sigma_2$.
\end{definition}
Note that since $I$ and $I^*$ partition the nontrivial indices, we have  $\max \{\sigma(I), \sigma(I^*) \}=\sigma_2$, so an $I$-subsystem is dominant iff $\sigma(I) \geq \sigma(I^*)$.

\begin{theorem}[Upper bound on quantized correlation]\label{thm:MainResultSubsystems}
For dominant $I$, $f\in \mathcal{F}_M(P_X)$ and $g\in \mathcal{F}_N(P_Y)$,
\begin{align}
    \langle f,g \rangle  \leq
    \Gfunc \left(\phi_M (H_{P_X}(I)), \phi_N(H_{P_Y}(I)) \right)= 
    \Gfunc \left(\sqrt{1-D_M(H_{P_X}(I))},\sqrt{1-D_N(H_{P_Y}(I))} \right), \label{eq:mainThmPartI}
\end{align}
where $\Gfunc$ is given by \eqref{eq:dfunc} with $(s,t)=(\sigma(I),\sigma(I^*))$, and consequently
\begin{align}
    \rho_{M,N}(X;Y) \leq \Gfunc  \left(\sqrt{1-D_M(H_{P_X}(I))},\sqrt{1-D_N(H_{P_Y}(I))} \right). \label{eq:thm1eq2}
\end{align}    
\end{theorem}

\begin{proof}
By Lemma~\ref{Lemma:BasicFullPaper} and Definition~\ref{def:maxproj}, the pair $(||f_I||_2, ||g_I||_2)$ satisfies $||f_I||_2 \leq \phi_M(H_{P_X}(I))$ and $||g_I||_2 \leq \phi_N(H_{P_Y}(I))$. Dominance gives $\sigma(I) \geq \sigma(I^*)$. Overall, Lemma~\ref{Lemma:BasicFullPaper} reads $\langle f, g \rangle \leq r\left( ||f_I||_2, ||g_I||_2 \right)$ with $(s,t) = (\sigma(I), \sigma(I^*))$. The definition of the supremum with $(a,b) = ( \phi_M(H_{P_X}(I)),  \phi_N(H_{P_Y}(I)))$ gives the inequality in~\eqref{eq:mainThmPartI}. The equality in~\eqref{eq:mainThmPartI} follows from Theorem~\ref{thm:ProjQuantDuality}, applied once with $\pi=P_X,\  H=H_{P_X}(I)$ and once with $\pi=P_Y,\  H=H_{P_Y}(I)$. Both applications are admissible since $I$ is dominant and excludes the trivial index, hence the subspaces are nonzero and orthogonal to the constant functions.
Finally, taking the supremum over $\mathcal{F}_M(P_X) \times \mathcal{F}_N(P_Y)$ and recalling Definition~\ref{def:quant_maxcorr} yields~\eqref{eq:thm1eq2}.
\end{proof}

To finalize the framework, the following technical lemma provides the explicit form of $d(a,b)$, and will prove helpful in the sequel, where we discuss concrete case studies. The derivation appears in the appendix.
\begin{lemma}\label{lem:simplify_func}
    Let $0\leq t\leq s$. We have that
\begin{align}\label{eq:SimplifyLemmaCases}
    d(a,b)=\begin{cases}
        \sqrt{b^2\cdot s^2+(1-b^2)\cdot t^2},& \{b> a\cap x^*\leq a \cap y^*\leq b\}\cup\{x^*\leq a\cap y^*> b\} , \\ \sqrt{a^2\cdot s^2+(1-a^2)\cdot t^2},& \{b\leq a\cap x^*\leq a\cap y^*\leq b\}\cup \{x^*> a\cap y^*\leq b\} ,\\ ab\cdot s+\sqrt{(1-a^2)(1-b^2)}\cdot t,& \text{o.w.},
        \end{cases}
\end{align}
where 
\begin{align}
   x^*=\frac{b\cdot s}{\sqrt{b^2\cdot s^2+(1-b^2)\cdot t^2}}, \hspace{2mm} y^*=\frac{a\cdot s}{\sqrt{a^2\cdot s^2+(1-a^2)\cdot t^2}} \label{eq:maximizer} .
\end{align}
\end{lemma}

\subsection{Product space}\label{subsec:productspace}
%We provide a relation between quantized maximal correlation and the MMSE in the problem of $M$-level quantization of normalized sums of i.i.d. random variables.

%%%%%%%%%%%%%%%%%%%%%%%%%%%%%%%%%%%%%%
%    Spectral theory high dim        %
%%%%%%%%%%%%%%%%%%%%%%%%%%%%%%%%%%%%%%
One of the most attractive features of the spectral decomposition of the canonical system (Section~\ref{subsec:spectralDecomp}) is its tensorization property. Let $(X^n,Y^n)\sim P_{XY}^{\otimes n}$, and for any vector $u=(u_1,...,u_n)\in \mathbb{N}^n$ define the \textit{product} functions
\begin{align*}
f_u(X^n)=\prod_{i=1}^n f_{u_i}(X_i), \quad g_u(Y^n)=\prod_{i=1}^n g_{u_i}(Y_i).
\end{align*}
It is easy to verify (see, e.g.,~\cite{folland1999modern}, exercise 5.61)  that $\{f_u\}_{u\in\mathbb{N}^n}$ and $\{g_u\}_{u\in\mathbb{N}^n}$ are orthonormal sets that span $L^2(P_X^{\otimes n})$ and $L^2(P_Y^{\otimes n})$ respectively. Thus, any two functions $f\in L^2(P_X^{\otimes n})$ and $g\in L^2(P_Y^{\otimes n})$ can be written as
\begin{align}\label{eq:SpectralFormHighDim}
    f(X^n)=\sum_{u\in \mathbb{N}^n} a_u f_u(X^n),\quad g(Y^n)=\sum_{u\in \mathbb{N}^n} b_u g_u(Y^n), 
\end{align}
where $a_u=\langle f(X^n),f_u(X^n)\rangle$ and $b_u=\langle g(Y^n),g_u(Y^n)\rangle$. Letting $\sigma_u = \prod_{i=1}^n \sigma_{u_i}$, we have
\begin{align*}
   \langle f_u(X^n),g_{v}(Y^n)\rangle&=\left\langle \prod_{i=1}^n f_{u_i}(X_i),\prod_{i=1}^n g_{v_i}(Y_i)\right\rangle \\&=\prod_{i=1}^n \langle f_{u_i}(X_i),g_{v_i}(Y_i)\rangle\\&=\prod_{i=1}^n \sigma_{u_i} \mathds{1}(u_i=v_i) \\&= \sigma_u \mathds{1}(u=v).
\end{align*}
Therefore, in the canonical spectral decomposition of the product distribution, the largest singular value is always $1$ and has multiplicity $1$, corresponding to $u=v=(1,1,\ldots,1)$. The second largest singular value is $\sigma_2$ and has multiplicity $n$, where the $n$ occurrences of $\sigma_2$ correspond to the $n$ singular functions $f_{{u(i)}}(X^n) = f_2(X_i)$ and $g_{{u(i)}}(Y^n) = g_2(Y_i)$, for  $1 \leq i \leq n$, where $u(i)$ has its $i$-th coordinate equal to $2$ and all other entries equal to 1, i.e., $u(i)=\mathds{1}+e_i$. Furthermore, the third largest singular value is $\max \{\sigma_3,\sigma_2^2\}$, and is attained by the vectors $u=(2,2,1,\ldots,1)$ or $u=(3,1,\ldots,1)$, up to permutations.  
\\
We can now write $\E[f(X^n)g(Y^n)]$ as
\begin{align}
   \langle f(X^n),g(Y^n)\rangle= \left\langle \sum_{u\in \mathbb{N}^n} a_u f_u(X^n), \sum_{v\in \mathbb{N}^n} b_{v} g_{v}(Y^n)\right\rangle=\sum_{u\in \mathbb{N}^n} a_u b_u \sigma_u.\label{eq:linear_form}
\end{align}
Since $\sum a_u^2=1$ for any $f$ with $\E[f(X)]=0,\Var(f(X))=1$, it is now easy to see that $\rho_m(X^n;Y^n)=\rho_m(X;Y) = \sigma_2$, and also that the scalar functions $f_2(X_k), g_2(Y_k)$ achieve the maximal correlation for any $k\in [n]$.

To make this formal, we cast the product structure in the language of Section~\ref{sec:spect}. The results there apply verbatim to the canonical system of $P_{XY}^{\otimes n}$, with index universe $\mathbb{N}^n$ and trivial index $(1,\ldots,1)$. Consider the collection of index sets
\begin{align}
    I_2 \triangleq \{u_2(1),\ldots,u_2(n)\}=\{\mathds{1}+e_1,\ldots,\mathds{1}+e_n\},
\end{align}
whose subsystem consists of the $n$ coordinate copies of the scalar pair $(f_2, g_2)$. We write
$\mathcal{H}_{f_2}^n \triangleq \overline{\mathrm{span}} \{ f_{u_2(1)} , \ldots ,f_{u_2(n)} \}$ and $\mathcal{H}_{g_2}^n$ is defined
analogously. By the multiplicity computation above, $H_{P_X^{\otimes n}}(I_2) = \mathcal{H}_{f_2}^n$ and $\sigma(I_2)=\sigma_2$, while the complement carries $\sigma_* \triangleq\sigma(I_2)^*=\max\{\sigma_3,\sigma_2^2\}$. Since $\sigma_2 \geq \sigma_3$ and $\sigma_2 \geq \sigma_2^2$, the $I_2$-subsystem is dominant, thus applying Theorem~\ref{thm:MainResultSubsystems} to the canonical system of $P_{XY}^{\otimes n}$ with dominant set $I_2$ we have the following upper bound on the quantized maximal correlation as a function of the optimal $M$-level MMSE of $\mathcal{H}_{f_2}^n$ and the optimal $N$-level MMSE of $\mathcal{H}_{g_2}^n$.

\begin{corollary}[Upper bound on quantized correlation in product space]
\label{cor:rhoasmmse} For $(s,t)=(\sigma(I_2),\sigma((I_2)^*))$ in \eqref{eq:dfunc},
\begin{align}
    \rho_{M,N}(X^n;Y^n) \leq \Gfunc \left(\sqrt{1-D_M(\mathcal{H}_{f_2}^n)},\sqrt{1-D_N(\mathcal{H}_{g_2}^n)}\right).
    \label{eq:RhoUpperBoundMMSE}
\end{align}
\end{corollary}
Assuming there is no cardinality constraint on $Y$, i.e., $N=\infty$, we have $D_\infty(\mathcal{H}_{g_2}^n)=0$ as there is no quantization loss, thus Corollary~\ref{cor:rhoasmmse} reduces to $  \rho_{M,\infty}(X^n;Y^n)\leq d(\sqrt{1-D_M(\mathcal{H}_{f_2}^n)},1)$.
This implies we should set $b=1,x^*=1$ in the upper bound of Lemma~\ref{lem:simplify_func}, which results in the following corollary, which quantifies the minimal correlation loss w.r.t. maximal correlation that arises from restricting $X$ to take only $M$ values, in terms of the MMSE loss. 
\begin{corollary}[Upper bound on correlation under single-side quantization]\label{cor:delta}
   It holds that
   \begin{align}
    (  \rho_{M,\infty}(X^n;Y^n))^2\leq \sigma_2^2-\Delta \cdot D_M(\mathcal{H}_{f_2}^n)=\rho_m^2(X,Y)-\Delta \cdot D_M(\mathcal{H}_{f_2}^n),
   \end{align}
   where $\Delta = \sigma_2^2-\sigma_*^2$. Similarly, $(  \rho_{\infty,N}(X^n;Y^n))^2\leq \rho_m^2(X,Y)-\Delta \cdot D_N(\mathcal{H}_{g_2}^n)$.
\end{corollary}

We conclude the section by noting that $D_M(\mathcal{H}_{f_2}^n)$ can be identified with a scalar quantization problem. Since $f_2(X_1),\ldots,f_2(X_n)$ are i.i.d.\ copies of the scalar source $f_2(X)$, and are orthonormal in $L^2(P_X^{\otimes n})$, the unit sphere of $\mathcal{H}_{f_2}^n$ consists exactly of the normalized linear combinations $\sum_{i=1}^n a_i f_2(X_i)$, $\|a\|_2=1$. Hence $D_M(\mathcal{H}_{f_2}^n)$ is the distortion in $M$-level quantization of a normalized linear combination of i.i.d.\ copies, which we study in Sections~\ref{sec:ratedist} and~\ref{sec:anticont}. Surprisingly, the maximal projection is exactly the variance of the optimal MMSE estimator.

% **********************************
% Gaussian Chapter
% **********************************

\section{A Gaussian Example}\label{sec:gaussian}
In most cases, finding the optimal $M$-level MMSE of subspace $\mathcal{H}_{f_2}^n$ is hard. However, when the subspace is such that the optimal quantizer is known, as in the Gaussian case below, Corollary~\ref{cor:rhoasmmse} can lead to an upper bound on the quantized correlation that is almost tight.

\begin{example}\label{ex:gaussian}
    Let $(X^n,Y^n)$ be i.i.d. pairs of  $\rho$-correlated jointly Gaussian random variables with expectation zero and unit variance, and assume w.l.o.g. that $\rho>0$. Then
    $$\rho_b(X^n;Y^n) \leq \frac{2}{\pi}\rho + \left(1-\frac{2}{\pi}\right)\rho^2.$$    
\end{example}
The assumption $\rho>0$ is without loss of generality since we can define $\tilde{Y}=-Y$ and consider $\rho_b(X^n;\tilde{Y}^n)$. Note that, as $P_X=P_Y$ and $M=N=2$, the two arguments of the function $d$ in the upper bound of Corollary~\ref{cor:rhoasmmse} coincide, so we evaluate its right-hand side via Lemma~\ref{lem:simplify_func} with $a=b$ and $(s,t) = (\sigma_2, \sigma_*)$, where $\sigma_*=\max \{\sigma_3, \sigma_2^2\}$ for tensor products, as established in Section~\ref{subsec:productspace}. It is easy to check from eq.~\eqref{eq:maximizer} that whenever $a=b$, we have $x^*=y^*\geq a$, implying the upper bound
\begin{align}
   \rho_b(X^n;Y^n)\leq  (1-D_2(\mathcal{H}_{f_2}^n))\sigma_2+D_2(\mathcal{H}_{f_2}^n)\sigma_*.\label{eq:same_marginals}
\end{align}
For the jointly Gaussian case with $\rho>0$, it was shown in~\cite{LancasterComplete1958} that $f_2(X)=X$ and $g_2(Y)=Y$, and the singular values are $\sigma_i = \rho^{i-1}$. 
Moreover, the MMSE $1$-bit quantizer of a standard Gaussian random variable is $Q(X) = \sqrt{\frac{2}{\pi}} \cdot\mathrm{sign}(X)$, implying that $D_2(\mathcal{H}_{f_2})=1-\frac{2}{\pi}$, establishing the claim for the $n=1$ case. For general $n$, note that any unit norm linear combination of i.i.d. standard Gaussian r.v.s is itself a standard Gaussian, so it admits the same distortion  $D_2(\mathcal{H}_{f_2}^n)=1-\frac{2}{\pi}$, which yields the stated bound.
It is however known from the work of Borell~\cite{borell1985geometric} that in the Gaussian case, the binary maximal correlation tensorizes and is achieved by a one-dimensional threshold function, implying $\rho_b(X^n;Y^n) = \frac{2}{\pi} \arcsin \rho = \frac{2}{\pi} \rho + O(\rho^3)$ by Sheppard formula. This agrees with our bound to first order.  In Figure~\ref{fig:Gaussian} we compare our upper bound with the exact value, attained by the threshold functions $\mathrm{sign}(X)$, $\mathrm{sign}(Y)$.

\begin{figure}[ht]
    \centering    \includegraphics[scale=0.30]{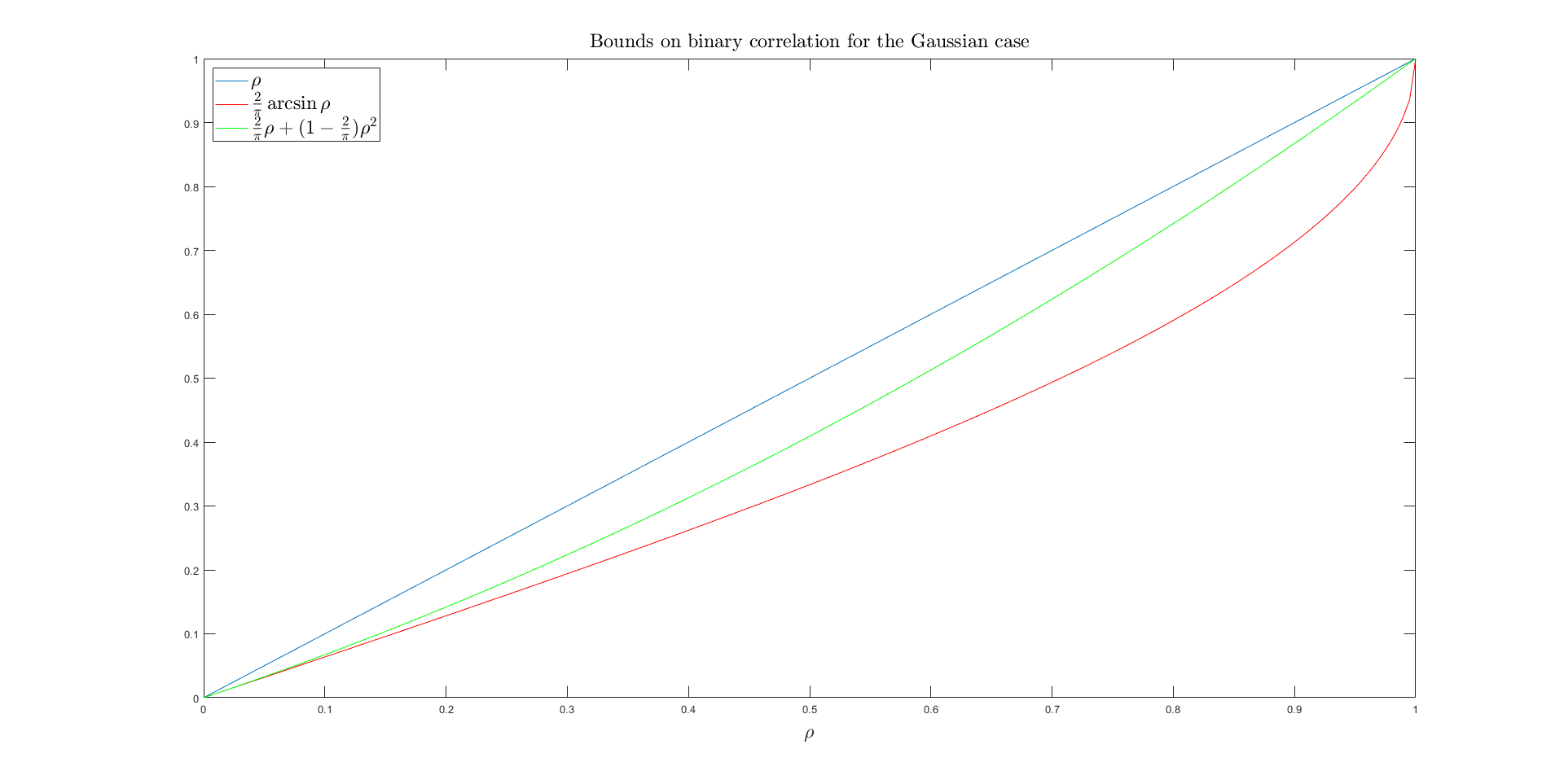}
    \caption{\label{fig:Gaussian} Bounds on the binary maximal correlation in the Gaussian case: the trivial bound $\rho$, our upper bound $\frac{2}{\pi}\rho+(1-\frac{2}{\pi})\rho^2$, and the exact value $\frac{2}{\pi}\arcsin\rho$.}
\end{figure}

% **********************************
% RD Chapter
% **********************************

\section{MSE Lower Bounds on quantization of linear combinations of i.i.d. random variables - \\The rate distortion approach}\label{sec:ratedist}
In Corollary~\ref{cor:rhoasmmse} we derived an upper bound on the $(M,N)$-quantized maximal correlation as a function of two quantities: the minimal MSE that can be achieved in reconstructing a unit-norm linear combination of $f_2(X_1),\ldots,f_2(X_n)$ from its $M$-value representation and, similarly, the minimal MSE in reconstructing a unit-norm linear combination of $g_2(Y_1),\ldots,g_2(Y_n)$ from its $N$-value representation. The next sections are dedicated to deriving MSE lower bounds on representations of linear combinations of random variables using $M$-level quantizers. This section leverages classic rate distortion results (i.e., source coding) from information theory to obtain lower bounds on the MSE in quantization  of linear combinations. 
In the following, $U_1,\ldots,U_n$ is a sequence of $n$ i.i.d random variables over alphabet $\mathcal{U}$ with zero mean and unit norm, and 
\begin{align}
   \phi_{a}(U^n) \triangleq a^TU^n=\sum_{i=1}^n a_i U_i 
\end{align} represents a linear combination of $U^n$ with vector weights $a$. Furthermore, let $\phi_{M}$ be an $M$-level quantized version of $\phi_{a}$. 

\begin{lemma}\label{lem:EPI}
    Let $U^n$ be an independent random vector with densities. Then for any $a \in \mathbb{R}^n$, it holds that
    \begin{align}
        h\left( \phi_{a}(U^n)\right) \geq \min_i h(U_i) + \log ||a||_2.
    \end{align}

\end{lemma}
\begin{proof}
    By entropy power inequality (Theorem 17.7.3 in~\cite{cover2012elements}), we have
    \begin{align*}
        2^{2h\left(\sum a_i U_i\right)} & \geq \sum 2^{2h(a_i U_i)}      \\
             & = \sum 2^{2(h(U_i)+\log|a_i|)} \\
            & = \sum a_i^2 2^{2h(U_i)}       \\
            & \geq 2^{2\min_i h(U_i)} \sum a_i^2.
    \end{align*}
Taking logarithm on both sides completes the proof.
\end{proof}
As in our case $\{U_i\}$ is an i.i.d. sequence and $||a||_2^2=1$, the lemma implies $h( \phi_{a}(U^n)) \geq h(U)$, where $U$ is a random variable with the same distribution as $U_i$. For brevity, from now on we denote $\phi(U^n)=\phi_a(U^n)$. 

\begin{definition}[Optimal $M$-level MMSE of a linear combination]\label{def:DnM}
%\dror{\textbf{Or} - Please approve of this definition: a) changes in display (elevated this instantiation of $D_M^n$ to *the* formal definition, per the generalization of III (In short: compare with the more general Definition~\ref{def:MMSEopt}, and connect the two via~\eqref{eq:bridge}) b) changing $\min$ to $\inf$ c) domain mismatch: $\mathcal{X}$ carried over from previous chapters). quickest repair: change the definition's domain and don't touch the chapter: I'll adjust my calls to this def inasmuch as I refer to it in VIII. Either way, you should approve of any changes to this chapter w.r.t itself}

\begin{align}
D_M^n(U)\triangleq \inf_{\substack{\fM:\mathcal{U}^n\to \{c_1,\ldots,c_M\} \\ a\in \mathbb{R}^n, \|a\|_2=1}}\E\left(f(U^n)-\sum_{i=1}^n a_i U_i\right)^2. 
\end{align}
\end{definition}
\begin{remark}
  Note that the assignment $U=f_2(X)$ when $f_2$ is w.r.t. some joint distribution $P_{XY}$, gives rise to the identity \begin{align}
   D_M^n(f_2(X))=D_M(\mathcal{H}_{f_2}^n).
\end{align}  
\end{remark}

% Let
% \begin{align}
% D_M^n(U)\triangleq \min_{\substack{\fM:\mathcal{X}^n\to \{c_1,\ldots,c_M\} \\ a\in \mathbb{R}^n\|a\|_2=1}}\E\left(f(X^n)-\sum_{i=1}^n a_i U_i\right)^2. 
% \end{align}
\begin{theorem}\label{thm:Rate_dist_bound}
Let $U^n\sim f_U^{\otimes n}$ and let $\mathsf{N}(U)=\frac{1}{2\pi e}2^{2h(U)}$ be the \textit{entropy power} of $U$. It holds that
\begin{align}
    D_M^n(U) \geq \frac{\mathsf{N}(U)}{M^2}.
\end{align}
\end{theorem}
\begin{proof} 
For any $a\in \mathbb{R}^n$ and any $M$-level quantizer $\phi_M=\phi_M(\phi_a)$, we have
\begin{align}
\nonumber    \log M &\geq I(\phi(U^n); \phi_M) \\& = h(\phi(U^n)) - h(\phi(U^n)|\phi_M)\\
 & = h(\phi(U^n)) - h(\phi(U^n)-\phi_M|\phi_M)\label{eq:dif_cons}\\
& \geq h(\phi(U^n)) - h(\phi(U^n)-\phi_M)\label{eq:conditioning}\\ &\geq h(\phi(U^n)) - \frac{1}{2}\log (2\pi e D_M^n(U))\label{eq:DistBoundPt2}, 
\end{align}
where~\eqref{eq:dif_cons} follows since $h(X+c)=h(X)$,~\eqref{eq:conditioning} follows since conditioning reduces differential entropy, and~\eqref{eq:DistBoundPt2} follows since the maximal differential entropy under a second-moment constraint is attained by a Gaussian distribution (Theorem 8.6.5 in~\cite{cover2012elements}), i.e., $\max_{\substack{\E(Y^2)\leq D}} h(Y) =  \frac{1}{2}\log (2\pi e D)$. Combining \eqref{eq:DistBoundPt2} and Lemma~\ref{lem:EPI}, we get
\begin{align}
    D_M^n(U) \geq \frac{1}{2 \pi e}2^{2(h(\phi(U^n))-\log M)}\geq \frac{1}{2 \pi e}2^{2(h(U)-\log M)}.
\end{align}
\end{proof}
Theorem~\ref{thm:Rate_dist_bound} provides a lower bound for any i.i.d. continuous r.v.s with densities $f_U$. In order to extend its result to discrete random variables with p.m.f. $P_U$ as well, we convert a discrete distribution to a continuous one by adding a random noise $Z$ supported on a small enough interval such that $U$ is still recoverable from $U+Z$. For any discrete $U$ supported over the alphabet $\mathcal{U}$, we define $d_{\min}$ as the smallest distance between any two letters in $\mathcal{U}$, that is,
\begin{align}
  d_{\min}=\underset{\underset{u_i\neq u_j}{u_i, u_j\in \mathcal{U}} }{\min}|u_i-u_j|.  
\end{align}
If we now add a continuous noise $Z$ supported on $(-d_{\min}/2,d_{\min}/2)$ to the source $U$, it is guaranteed that $U$ can be recovered without loss from $U+Z$. This gives rise to the following result.

\begin{theorem}\label{thm:distort_discrete_case}
    Let $U^n\sim P_U^{\otimes n}$ and let $Z$ be some continuous random variable supported on $(-d_{\min}/2,d_{\min}/2)$ with zero mean and variance $\sigma^2$. It holds that
    \begin{align}
   D_M^n(U) \geq \sigma^2\left(\frac{2^{2H(U)}}{M^2}\cdot\frac{\mathsf{N}(Z)}{\sigma^2}-1\right).
\end{align}
\end{theorem}

\begin{proof}

For any $a\in \mathbb{R}^n$ and any $M$-level quantizer $\phi_M=\phi_M(\phi_a)$, we have the following chain of (in)equalities:
\begin{align}
    \log M &\geq I\left(\phi(U^n);\phi_M\right) \\&\geq I\left(\sum _{i=1}^n a_i (U_i+Z_i);\phi_M\right)\label{eq:DFI}
    \\&= h\left(\sum _{i=1}^n a_i (U_i + Z_i)\right) - h\left(\sum _{i=1}^n a_i(U_i + Z_i)\middle|\phi_M\right)\\&= h\left(\sum _{i=1}^n a_i (U_i + Z_i)\right) - h\left(\sum _{i=1}^n a_i(U_i + Z_i)-\phi_M\middle|\phi_M\right)\\&\geq h\left(\sum _{i=1}^n a_i (U_i + Z_i)\right) - h\left(\sum _{i=1}^n a_i(U_i + Z_i)-\phi_M\right)\\
    &\geq h\left(\sum _{i=1}^n a_i (U_i + Z_i)\right) - \frac{1}{2}\log \left(2\pi e \E \left(\sum _{i=1}^n a_i U_i - \phi_M + \sum _{i=1}^n a_i Z_i \right)^2 \right)\label{eq:Gaussian}\\
&\geq h\left(\sum _{i=1}^n a_i (U_i + Z_i)\right) - \frac{1}{2}\log (2\pi e(D_M^n(U) + \sigma^2 ||a||_2^2) )\label{eq:disto}\\
&\geq h(U+Z)  - \frac{1}{2}\log (2\pi e(D_M^n(U) + \sigma^2 ) )\label{eq:EPIClaim}
\\& =H(U)+h(Z)- \frac{1}{2}\log \left(2\pi e\left(D_M^n(U) + \sigma^2 \right), \right)\label{eq:sum_ent}
\end{align}
where~\eqref{eq:DFI} follows from the data processing inequality since $\phi_M - \phi(U^n) -  (\phi(U^n) + \sum a_i Z_i)$ forms a Markov chain,~\eqref{eq:Gaussian} follows since differential entropy is maximized by a Gaussian r.v. under a second moment constraint,~\eqref{eq:disto} follows from $\E(\phi(U^n)-\phi_M)^2\leq D_M^n(U)$,~\eqref{eq:EPIClaim} follows from Lemma~\ref{lem:EPI}, and~\eqref{eq:sum_ent} follows from $h(U+Z) = I(U;U+Z) + h(Z) = H(U)+h(Z)$. The result is achieved by rearranging terms.
\end{proof}
By choosing a specific distribution on $Z$, Theorem~\ref{thm:distort_discrete_case} above admits a more elegant form.
\begin{corollary}\label{cor:unif}
For $U^n\sim P_U^{\otimes n}$, we have
    \begin{align}
D_M^n(U) \geq \frac{d_{\min}^2}{12}
\left(\frac{12}{2\pi e}\cdot\frac{2^{2H(U)}}{ M^2}-1\right)
\end{align}
\end{corollary}
The corollary follows by setting the distribution of $Z$ to be uniform over $(-d_{\min}/2,d_{\min}/2)$, which implies that $\sigma^2=\frac{d_{\min}^2}{12}$ and $h(Z)=\log (d_{\min})$.
\begin{corollary}
For $U^n\sim P_U^{\otimes n}$, we have
\begin{align}
D_M^n(U) \geq \max_{ \sigma_Z \leq \frac{1}{12}d_{\min}^2} \sigma^2 \left[(2\Phi(\alpha) - 1)^2 2^{ - \frac{2\alpha\varphi(\alpha)}{2\Phi(\alpha) - 1}}\cdot\frac{2^{2H(U)}}{M^2} - 1\right] 
\end{align}
where  
$\alpha=\frac{d_{\min}}{2\sigma}$, and $   \sigma_Z=\sigma\left(1-\frac{2\alpha\cdot \varphi(\alpha)}{2\Phi(\alpha) - 1}\right)$, where $\varphi$ and $\Phi$ are the PDF and CDF of a standard Gaussian distribution, respectively. 
\end{corollary}

\begin{proof}
Let $N$ be a normal random variable with zero mean and variance $\sigma_N$, and define $Z$ as $N$ conditioned on the event $N\in [-d_{\min}/2,d_{\min}/2]$, which is known as the truncated normal distribution over the support $[-d_{\min}/2,d_{\min}/2]$, with zero mean and variance $\sigma^2$, where (see, e.g.,~\cite{del1994singly}),
\begin{align}
    f_Z(z) = \frac{1}{\sigma}\frac{\varphi(\frac{z}{\sigma})}{2\Phi(\alpha)-1}\cdot \mathbbm{1}  \left\{|z|\leq \frac{d_{\min}}{2}\right\},\hspace{2mm}    \sigma_Z=\sigma\left(1-\frac{2\alpha\cdot \varphi(\alpha)}{2\Phi(\alpha) - 1}\right).
\end{align}
The truncated normal distribution is chosen here since it maximizes $h(Z)$ under the constraints of a fixed variance and fixed support (as long as $\sigma_Z \leq \frac{1}{12}d_{\min}^2$). The proof follows by substituting~\cite{johnson1994continuous}
\begin{align}
    h(Z) = \frac{1}{2}\log \left(2\pi e \sigma^2 (2\Phi(\alpha) - 1)^2\right) -  \frac{\alpha\varphi(\alpha)}{2\Phi(\alpha) - 1}.
\end{align}
\end{proof}
\begin{remark}
The lower bound of Theorem~\ref{thm:distort_discrete_case} is tightest when $Z$ is chosen to be the random variable that maximizes the differential entropy under the constraints $Z\in[-\frac{d_{\min}}{2},\frac{d_{\min}}{2}]$ w.p. $1$ and $\Var(Z)=\sigma^2$. The pdf of the optimal $Z$ has the form 
\begin{align}
   f_Z(z)=\frac{e^{\lambda z^2}}{\int_{-\frac{d_{\min}}{2}}^{\frac{d_{\min}}{2}}e^{\lambda z^2}dz} \cdot \mathbbm{1}  \left\{|z|\leq \frac{d_{\min}}{2}\right\}.
\end{align}
It was shown in~\cite{dowson1973maximum} that for $0\leq \sigma^2\leq \frac{1}{12}d_{\min}^2$, $\lambda$ is negative and $Z$ has the truncated Gaussian distribution, whereas for $\frac{1}{12}d_{\min}^2< \sigma^2\leq \frac{1}{4}d_{\min}^2$, $\lambda$ is positive and $Z$ has the so called \textit{truncated U} distribution (note that $\sigma^2>\frac{1}{4}d_{\min}^2$ is unattainable). Due to the cumbersome expressions corresponding to the truncated U distribution, the bound we present here only maximizes over the truncated Gaussian distribution.
\end{remark}

The bound of Theorem~\ref{thm:distort_discrete_case} might be improved by increasing the minimum distance of the alphabet. This can be accomplished by artificially introducing a random variable $A$ that depends on $U$, but such that $(U,A)$ is independent of $Z$ and $U$ can be recovered from $(U+Z,A)$. Let $\mathcal{U}_a\subset \mathcal{U}$ denote the support of $U$ conditioned on $A=a$, and let $Z$ be a continuous random variable supported over an interval of size
\begin{align}
  d_{\min}({P_{A|U}})=\min_{a\in A} d_{\min}(\mathcal{U}_a). 
\end{align}
Note that $Z$ does not dependent on any realization of $(U,A)$, only on the conditional law $P_{A|U}$. The bound derived below provides an improvement whenever one can find a random variable $A$ for which $H(U|A)\approx H(U)$, and also $d_{\min}({P_{A|U}})$ is large.
\begin{theorem}[Improved Rate Distortion]\label{thm:hz_bound}
Assume the random variables $(U,A,Z)$ satisfy
\begin{enumerate}
    \item $(U,A,Z)\sim P_UP_{A|U}P_Z$
    \item $Z\sim f_Z$ is supported on $(-d_{\min}({P_{A|U}})/2,d_{\min}({P_{A|U}})/2)$ and has zero mean and variance $\sigma^2$
    \item $H(U|U+Z,A)=0$
\end{enumerate}
Then it holds that
\begin{align}
D_M^n(U) \geq \sigma^2\left(\frac{2^{2(H(U)-I(U;A))}}{M^2}\cdot\frac{\mathsf{N}(Z)}{\sigma^2}-1\right).
\end{align} 
\end{theorem}
\begin{proof}
  The proof follows immediately from eq.~\eqref{eq:EPIClaim} by writing 
\begin{align}
  h(U+Z)&\geq h(U+Z|A)\\&=I(U;U+Z|A)+h(Z|A)\\&=H(U|A)-H(U|U+Z,A)+h(Z|A)\\&=H(U)-I(U;A)+h(Z)\label{eq:last},
\end{align}
where in~\eqref{eq:last} we used the fact that $H(U|U+Z,A)=0$.   
\end{proof}
\begin{example}
  Consider the case of $U \sim \mathrm{Unif}([K]\cup{K+\eps})$ for some $0<\eps<1$. The minimal distance of the symbol space is $\eps$ and is dictated by the symbol ${K+\eps}$. Thus, if we define $A=\mathds{1}(U=K+\eps)$ we get that $\mathcal{U}_0=[K]$, and $\mathcal{U}_1=\{K+\eps\}$, so that $d_{\min}^A=1$. 
  Taking $Z$ to be uniform over $(-d_{\min}/2,d_{\min}/2)$ independently of $(U,A)$, and noting that $U$ can be recovered from $U+Z$ if $A$ is also known, we have 
    \begin{align}
D_M^n(U) \geq \frac{1}{12}
\left(\frac{12}{2\pi e}\cdot\frac{2^{2H(U|A)}}{ M^2}-1\right)=\frac{1}{12}
\left(\frac{12}{2\pi e}\cdot\frac{(K^2)^{\frac{K}{K+1}}}{ M^2}-1\right)
\end{align}
as $H(U|A)=\frac{K}{K+1}\log K$, while Corollary~\ref{cor:unif} gives us 
    \begin{align}
D_M^n(U) \geq \frac{\eps^2}{12}
\left(\frac{12}{2\pi e}\cdot\frac{(K+1)^2}{ M^2}-1\right).
\end{align}
Thus, for constant $K$, we improve over Corollary~\ref{cor:unif} by a factor of $O(1/\eps^2)$.   
\end{example}
 
\section{MSE Lower Bounds on quantization of linear combinations of i.i.d. random variables - \\An anti-concentration approach}\label{sec:anticont}
In this section we formalize the intuition that if the probability of any union of $M$ small intervals is small, then any $M$-level quantizer must have large MSE, i.e., a large \textit{anti-concentration} implies a large distortion in $M$-level quantization.
We are thus interested in anti-concentration inequalities for random variables, that upper bound the largest probability of an interval of given length. Classical results from probability theory relate upper bounds on the anti-concentration function of sums of i.i.d. random variables to the anti-concentration function of a single random variable. As $D_M^n(U)$ measures how well a linear combination of $n$ random i.i.d. copies of $U$ can be quantized to $M$ levels, we can leverage these results for obtaining lower bounds on $D_M^n(U)$.

 Define the \textit{concentration function} of an r.v. $U$ as 
\begin{align*}
Q(U,t) \triangleq \sup_{a\in\mathbb{R}}\Pr \left(a \leq U \leq a+t \right), \quad t \geq 0.    
\end{align*}
It is easy to see that $t \mapsto Q(U,t)$ is non-decreasing and that  $Q(\alpha U,t)=Q(U,t/\alpha)$ for any $\alpha > 0$. 
%and $Q(X,\alpha t) \leq (1+\lfloor \alpha \rfloor) Q(X,t)$ 
Furthermore,~\cite{petrov2012sums} proved that if $U$ and $Z$ are independent, then
\begin{align}\label{eq:lemmaMinimumConcentration}
    Q(U+Z,t) \leq \min \{Q(U,t), Q(Z,t)\}.
\end{align}
As we are particularly interested in anti-concentration inequalities for sums of i.i.d. r.v.s, we define the \textit{$n$-fold concentration function} of $U$ to be 
    \begin{align}
        A_n(U,t) \triangleq \sup_{a\in \mathbb{R}^n: \|a\|_2 = 1} Q\left(\sum_{i=1}^n a_i U_i, t \right),\label{eq:OfersFunc_n}
    \end{align}
    where $U_1,\ldots,U_n$ are i.i.d copies of $U$. Next, we define the \textit{asymptotic concentration function} of $U$ to be 
    \begin{align}
        A(U,t) \triangleq \lim_{n \to \infty} A_n(U,t),\label{eq:OfersFunc}
    \end{align}
where the limit exists since $A_n(U,t)$ is a bounded non-decreasing function of $n$.
%%%%%%%%%%%%%%%%%%%%%%%%%%%%%%%%%%%
%      Ruldelson-Vershynin        %
%%%%%%%%%%%%%%%%%%%%%%%%%%%%%%%%%%%
Upper-bounding $A_n(U,t)$ is closely related to the Littlewood-Offord problem \cite{je1943number,erdos1945lemma}, which is concerned with the maximal possible value of $Q(\sum \alpha_i U_i,t)$ for $t = 0$, where $U_i$ are i.i.d $\mathrm{Ber}(1/2)$ and $|\alpha_i|\geq 1$. This problem and its variations have been extensively studied in additive combinatorics~\cite{tao2006additive}, however here we are mostly interested in the moderate $t$ regime. Furthermore, several works considered bounding $A(U,t)$ as a function of $Q(U_i,t)$ and related quantities, e.g., by Kolmogorov-Rogozin~\cite{kolmogorov1956two,rogozin1961increase} and Esseen~\cite{esseen1968concentration}. One such bound is due to \cite[Corollary 1.4]{rudelson2015small} and yields
\begin{align}
    A_n(U,t) \leq C_1 Q(U,t/2),\label{eq:RudeVer}
\end{align}
where $C_1 = \frac{12}{11}4\sqrt{2}$. For the purpose of this paper, we will rely both on \eqref{eq:RudeVer}, as well as on a bound of a different flavor that we establish based on the work of~\cite{esseen1968concentration}. Let $U^* = U - U'$ be the symmetrized version of $U$, where $U'$ is an independent copy of $U$. For any $t>0$, define 
\begin{align}
    E(U,t) \triangleq \E (\min\left\{U^2/t^2, 1\right)\}),\label{eq:defTruncatedVarD}
\end{align}
where $E(U,0)=\Pr(U \neq 0)$. The following lemma was proved in~\cite{esseen1968concentration}.
\begin{lemma}\label{lem:Petrov}
    Let $S = U_1 + \ldots + U_n$ for independent $U_1,\ldots,U_n$. 
    Then for any $0 < t_1,\ldots,t_n < t$, it holds that 
    \begin{align*}
        Q(S,t) \leq  \frac{Ct}{\sqrt{\sum_{i=1}^n t_i^2E(U^*_i;t_i)}},
    \end{align*}
    where $C=(96/95)^2 \sqrt{48\pi/11}$ and $U^*$ is the symmetrized r.v.~corresponding to $U$.
\end{lemma} %/Esseen
Lemma~\ref{lemma:DIMINDEP_A} gives upper and lower bounds on the concentration function based on $A\left(U,t_\tau \right)$, for $t_\tau$ defined below. The lower bound follows since $A_n\left(U,t_\tau \right)$ is non-decreasing in $n$, while the upper bound is a result of Lemma~\ref{lem:Petrov} (see Appendix~\ref{app:proof_lemma5}).
\begin{lemma}\label{lemma:DIMINDEP_A}
Let $t_\tau \triangleq Q(U,\tau) \sqrt{E(U^*;\tau)}C^{-1}\tau$. Then for any $n$,
\begin{align}
    A_n\left(U,t_\tau \right) \leq A\left(U,t_\tau \right) \leq Q(U,\tau).\label{eq:Lemma1}
\end{align}
\end{lemma}
\begin{theorem}\label{thm:Dbound3}
It holds that
\begin{align}
    D_M^n(U)\geq \max\{d_1(U),d_2(U)\},\label{eq:LinDistBound}
\end{align}
where
\begin{align}
    d_1(U)&=\max_t t^2\left|1 - M C_1 \cdot Q(U,t) \right|_+,\label{eq:anti_cont_bound}\\d_2(U)&=\max_t \frac{t^2}{4C^2}Q^2(U,t)E(U^*,t)|1-M Q(U,t)|_+.
\end{align}
\end{theorem}

\begin{proof}
Recall that $\phi(U^n) =\sum_{i=1}^n a_i U_i$ with $||a||_2^2=1$ and let $\phi_M$ be an $M$ level quantizer of $\phi(U^n)$ with output levels $m_1,\ldots,m_M$. Write
\begin{align}
\E (\phi(U^n) - \phi_M)^2   & \geq \Pr \left(|\phi
(U^n) - \phi_M| > t/2 \right)\cdot \left(t/2\right)^2\label{eq:Chebi}\\               & = \left(1-\Pr \left(|\phi(U^n) - \phi_M| \leq  t/2 \right)\right)\cdot t^2/4\\
& =\left( 1- \Pr \left(\phi(U^n) \in \bigcup\limits_{i=1}^{M} (m_i- t/2,m_i+ t/2) \right) \right) t^2/4
\\
& \geq \left( 1- \sum_{i=1}^M \Pr \left(\phi(U^n) \in (m_i- t/2,m_i+t/2) \right) \right)\cdot t^2/4\label{eq:union}\\
& \geq \left( 1 - M \max_{\alpha} \Pr \left(\phi(U^n) \in (\alpha,\alpha+t)\right) \right)\cdot t^2/4\\
& = \left(1 - M Q(\phi(U^n),t)\right)\cdot t^2/4,\label{eq:Qdef}\\
    &\geq \left(1 - M \cdot A_n(U,t) \right)t^2/4\label{eq:Adef}
\end{align}
where~\eqref{eq:Chebi} follows from Chebyshev's inequality,~\eqref{eq:union} follows from the union bound,~\eqref{eq:Qdef} from the definition of $Q(U,t)$ and~\eqref{eq:Adef} from the definition of $A_n(U,t)$. %Recall $Q(\alpha U,t)=Q(U,t/\alpha)$  write $t = ||a||_2\tau=\tau$. Then
The lower bound of $d_1(U)$ is attained by lower bounding~\eqref{eq:Adef} with~\eqref{eq:RudeVer}, and the lower bound of $d_2(U)$ is attained by lower bounding~\eqref{eq:Adef} using Lemma~\ref{lemma:DIMINDEP_A} (and then maximizing over $t)$.
\end{proof}
\section{Lower bounds examples}\label{sec:examples}
%This section is dedicated to a quantitative comparison of the lower bounds of sections~\ref{sec:ratedist} and~\ref{sec:anticont}, showing that no one is dominated by the other, but rather for different distributions and regimes, different approaches for lower bounds will be tight. Intuitively, one would expect that when the underlying distribution is well spaced out, the anti concentration bound might be tighter than the rate distortion bound, and vice versa. Below we translate this intuition into specific examples, both continuous and discrete, for which different bounds are preferable. 
In this section, we quantitatively compare the lower bounds developed in Sections~\ref{sec:ratedist} and~\ref{sec:anticont}. We show that neither bound dominates the other: depending on the underlying distribution, either approach may yield the tighter lower bound. Roughly speaking, the anti-concentration bound is expected to perform better for distributions whose mass is sufficiently spread out, while the rate-distortion bound is expected to be stronger for more concentrated distributions. We make this intuition precise by exhibiting both continuous and discrete examples in which each of the approaches is respectively tighter.
\subsection{Rate distortion beats anti-concentration}
\begin{example}[Continuous case]
 Let $U\sim \mathcal{N}(0,1)$. Then for $M=2$, the rate distortion bound of Theorem~\ref{thm:Rate_dist_bound} is $D_2^n(U)\geq 0.25$,
 while the anti concentration bound of Theorem~\ref{thm:Dbound3} only gives    $D_2^n(U)\geq c$,
 for some $c$ smaller than $0.0136$. This should be juxtaposed with the ground truth of $D_2^n(U)=1-\frac{2}{\pi}$, which is achieved by taking $f(U)=\sign(U)$ and reconstructing the conditional mean of $U$ given $f(U)$.
% \Or{For Gaussian case with $M=2$ we have $D_2^n(U)=1-\frac{2}{\pi}$. We teach in in IT, right? Add this to the example.}
\end{example}
In this example, the rate distortion bound is tighter as the Gaussian distribution is too concentrated for the anti-concentration approach to perform well. The rate distortion lower bound is $$D_M^n(U) \geq \frac{\mathsf{N}(U)}{M^2}=\frac{1}{M^2}$$ as the entropy power of a standard Gaussian r.v. is $1$. 
For the anti-concentration bound, we first evaluate $$d_1(U)=\max_t t^2\left|1 - M C_1 \cdot Q(U,t) \right|_+.$$ As the standard Gaussian distribution is symmetric with monotonically decreasing pdf for $u\geq 0$, the maximum of the concentration function $Q(U,t)=\sup_{a\in\mathbb{R}}\Pr \left(a \leq U \leq a+t \right)$ for any $t>0$ is achieved for $a=-\frac{t}{2}$.  This implies that 
$$Q(U,t)=\Phi\left(\frac{t}{2}\right)-\Phi\left(-\frac{t}{2}\right)=2\Phi\left(\frac{t}{2}\right)-1.$$ which in turn implies that $Q'(U,t)=\varphi\left(\frac{t}{2}\right)$. To evaluate $d_1(U)$, we maximize over $g(t)=t^2(1 - M C_1  Q(U,t))$:
\begin{align}
   g'(t)=2t(1 -M C_1 Q(U,t))-M C_1t^2Q'(U,t)&=0\\2 -2M C_1 \left(2\Phi\left(\frac{t^*}{2}\right)-1\right)- MC_1t^*\cdot\varphi\left(\frac{t^*}{2}\right)&=0\\ t^*\cdot\varphi\left(\frac{t^*}{2}\right)+4\Phi\left(\frac{t^*}{2}\right)-2\left(1+\frac{1}{MC_1}\right)&=0.
\end{align}
For $M=2$, a numeric computation gives $t^*\approx 0.14$ and $d_1(U)=g(t^*)\approx 0.0061$. Note that
$$d_2(U)=\max_t \frac{t^2}{4C^2}Q^2(U,t)E(U^*,t)|1-M Q(U,t)|_+$$
is also loose, as $E(U^*,t)\leq 1,Q^2(U,t)\leq 1$, and also $|1-2 Q(U,t)|_+\leq 1$ and is strictly positive only when $t\leq 1.4$, implying that  $\frac{t^2}{4C^2}\leq \frac{1.4^2}{4C^2}=0.0136$. Generally, the Lloyd-Max quantizer achieves the optimum for any $M$ in the Gaussian case, as it does for any log-concave distribution~\cite{fleischer1964sufficient}.
\begin{example}[Discrete case]
  Let 
\begin{align}
  \Pr(U=u)=\begin{cases}
\frac{1}{2},& u=0,\\
     \frac{1}{2L},& u=\pm d\cdot k \text{ and } 1\leq k \leq \frac{L}{2},  
\end{cases}  
\end{align}
where $d>0$ is chosen such that $\Var(U)=1$, that is, $d^2=\frac{12}{(L+1)(L/2+1)}$. Then for any $M$ and $L>\frac{\pi e M^2}{24}$, the rate distortion bound of Corollary~\ref{cor:unif} is $$D_M^n(U)\geq \frac{1}{(L+1)(L/2+1)}\left(\frac{24L}{\pi eM^2}-1\right),$$
 while the anti concentration only gives $D_M^n(U)\geq 0$. 
\end{example}
This is a highly concentrated discrete distribution where one dominant element has probability $1/2$, and $L$ other uniformly distributed values have probability $1/2L$. Note that $Q(U,t)\geq 1/2$ for any $t\geq 0$, implying that both $d_1(U)=0$ and $d_2(U)=0$ for any $M>1$, thus the anti-concentration lower bound is trivially zero. However, the rate distortion lower bound of Corollary~\ref{cor:unif} is strictly positive whenever $L>\frac{\pi e M^2}{24}$. To see this, note that $H(U)=\frac{1}{2}\log (2L)+\frac{1}{2}\log 2=\frac{1}{2}\log L+1$, implying that
$$\frac{12}{2\pi e}\cdot \frac{2^{2H(U)}}{M^2}-1=\frac{24L}{\pi eM^2}-1.$$
\subsection{Anti-concentration beats rate distortion}
\begin{example}[Continuous case]
   Let $X$ be a discrete r.v. uniformly distributed over a $L$-PAM constellation with $L=20$ and symbol distance
   \begin{align}
    d=\sqrt{\frac{12}{L^2-1}}\cdot \sqrt{1-\varepsilon}= 0.173\sqrt{1-\varepsilon}. \label{eq:sym_dis}  
   \end{align} That is, $$X\in \left\{-(L-1)\frac{d}{2},\ -(L-3)\frac{d}{2},\ \ldots,\ (L-3)\frac{d}{2},\ (L-1)\frac{d}{2}\right\},$$
   and has expectation zero and variance $1-\varepsilon$. Further let $Z\sim \mathcal{N}(0,\varepsilon)$ and define $U=X+Z$. Then for $M=2$ and $\varepsilon$ small enough, the rate distortion bound of Theorem~\ref{thm:Rate_dist_bound} is $D_2^n(U)\geq \Omega(\eps)$,
 while the anti concentration bound gives $D_2^n(U)\geq 0.01$ independent of $\eps$. 
\end{example}
Here $U$ is a continuous random variable with zero mean and unit norm, whose distribution has both low concentration, which makes the anti-concentration bound large, and low entropy power, which makes the rate distortion bound small. To evaluate the rate distortion lower bound, first note that $h(U)-h(U|X)=I(X;U)\leq H(X)$, thus 
\begin{align*}
  h(U)\leq H(X)+h(U|X)=\log(20)+\frac{1}{2}\log(2\pi e\varepsilon) =\frac{1}{2}\log(400\cdot 2\pi e\varepsilon).  
\end{align*}
Theorem~\ref{thm:Rate_dist_bound} then implies the lower bound $\frac{\mathsf{N}(U)}{2^2}=\frac{2^{2h(U)}}{8\pi e}\leq 100\varepsilon$, which can be made arbitrarily close to zero. 
To evaluate the anti concentration bound, note that the density function of $U$ is a convolution between a uniform discrete distribution over $20$ equispaced symbols and a narrow Gaussian envelope. This implies that $Q(U,d)= 1/20-\delta(\varepsilon)$ for any $t>0$, where $\delta(\varepsilon)\rightarrow 0$ as $\varepsilon\rightarrow 0$. Write
\begin{align}
    D_2^n(U)&\geq  \max\{d_1(U),d_2(U)\}\\&\geq\max_t t^2(1 - 2 C_1 \cdot Q(U,t))\label{eq:d_1bound}\\ &\geq d^2(1 - 12.342 \cdot Q(U,d))\label{eq:d_is_t}\\ &= d^2(1 - 12.342 \cdot (1/20-\delta(\varepsilon))\\&\geq 0.01\label{eq:ddef},
\end{align}
where eq.~\eqref{eq:d_1bound} follows since $d_1(U)\leq \max\{d_1(U),d_2(U)\}$, eq.~\eqref{eq:d_is_t} follows by setting $t=d$, and eq.~\eqref{eq:ddef} from substituting eq.~\eqref{eq:sym_dis} for $\varepsilon\ll 1$. This bound is better than the rate-distortion bound as it is independent of $\varepsilon$ (for $\varepsilon$ small enough).  
\subsection{Anti concentration beats rate-distortion but loses to improved rate distortion bound}
\begin{example}[Discrete case]
  Let $X$ be defined as in previous example where $L$ is some constant, that is, $X$ is an $L$-PAM constellation with symbol distance  $$d=\sqrt{\frac{12}{L^2-1}}\cdot\sqrt{1-\eps}.$$ Further let $Y$ be a binary r.v. that equals $\sqrt{\varepsilon}$ w.p. $1/2$ or $-\sqrt{\varepsilon}$ w.p. $1/2$, and define $U=X+Y$. Then for any $M<L/C_1$ and $\varepsilon\leq \frac{12}{L^2+11}$, the rate distortion bound of Corollary~\ref{cor:unif} is $D_2^n(U)\geq \Omega(\eps)$,
 while the anti concentration bound gives $$D_M^n(U)\geq \frac{12(1-\varepsilon)}{L^2-1}\left(1 - \frac{M C_1}{L}\right).$$
However, the improved rate distortion bound of Theorem~\ref{thm:hz_bound} gives $$D_M^n\geq \frac{1-\varepsilon}{L^2-1}\left(\frac{12}{2\pi e}\cdot\frac{L^2}{M^2} -1\right).$$
\end{example}
Here we construct a uniform distribution over $2L$ discrete symbols with small minimum distance, making the anti-concentration bound tighter than that of rate distortion. We then proceed to find a random variable $A$ such that given $A$, the minimum distance is significantly increased, making the rate distortion bound tighter than anti-concentration.  Note that $U$ is a zero mean unit variance random variable with $d_{\min}=2\sqrt{\varepsilon}$ for $\varepsilon \leq \frac{12}{L^2+11}$. The bound of Corollary~\ref{cor:unif} is
\begin{align}
   D_M^n(U) \geq \frac{d_{\min}^2}{12}
\left(\frac{12}{2\pi e}\cdot\frac{2^{2\log 2L}}{ M^2}-1\right)=\varepsilon \left(\frac{8}{\pi e}\cdot\frac{L^2}{ M^2}-1\right),
\end{align}
which is  $O(\varepsilon)$ for constant $L$ and $M$. The anti-concentration bound however gives
\begin{align}
    d_1(U)&=\max_t t^2(1 - M C_1 \cdot Q(U,t)) \\&\geq d^2(1 - M C_1 \cdot Q(U,d))\\ &\geq \frac{12(1-\varepsilon)}{L^2-1}(1 - M C_1 /L).\label{eq:compare_to_anti}
\end{align}
by setting $t=d$ and noting that $Q(U,d)=1/L$. We can obtain an improvement to rate distortion  bound by appealing to Theorem~\ref{thm:hz_bound}. Specifically, let $A=Y$. Then $A$ is dependent on $U$ and $d_{\min}({P_{A|U}})=d$ as $U|_{A=a}$ is simply a $L$-PAM constellation with minimum distance $d$ shifted by $a$. Now let $Z$ be a uniform random variable supported on $$(-d_{\min}({P_{A|U}})/2,d_{\min}({P_{A|U}})/2)=(-d,d).$$ Note that $U=X+Y$ is a deterministic function of $(U+Z,A)=(X+Y+Z,Y)$ as $X$ is recoverable from $X+Z$. We thus have
\begin{align}
D_2^n(U) &\geq \sigma^2\left(\frac{2^{2(H(U)-I(U;A))}}{M^2}\cdot\frac{\mathsf{N}(Z)}{\sigma^2}-1\right)\\&=\frac{d^2}{12}\left(\frac{12}{2\pi e}\cdot\frac{2^{2(\log 2L-1)}}{M^2} -1\right)\\&=\frac{1-\varepsilon}{L^2-1}\left(\frac{12}{2\pi e}\cdot\frac{L^2}{M^2} -1\right),
\end{align} 
which is greater than eq.~\eqref{eq:compare_to_anti} for $M<L/C_1$.

%\newpage
%%%%%%%%%%%%%%%%%%%%%%%%%%%%%%%%%%%
%         APPLICATIONS            %
%%%%%%%%%%%%%%%%%%%%%%%%%%%%%%%%%%%
\section{Applications}\label{sec:apps}

% % %%%%%%%%%%%%%%%%%%%%%%%%%%%%%%%%%%%
% % %     IMPROVED LOWER BOUND        %
% % %%%%%%%%%%%%%%%%%%%%%%%%%%%%%%%%%%%
\subsection{Improved lower bound on the probability of disagreement}
To lower bound the probability of disagreement between boolean functions of $X$ and $Y$, we can use Corollary~\ref{cor:rhoasmmse} to improve~\eqref{eq:Witsen3} (and therefore improve~\cite[Theorem 2]{Witsenhausen_Main}):
\begin{align*}
    \Pr(f(X^n)  \neq  g(Y^n)) \geq 2\sqrt{p(1-p)q(1-q)}\cdot \left(1-\Gfunc \left(\sqrt{1-D_2^n(f_2(X))},\sqrt{1-D_2^n(g_2(Y))}\right)\right).
\end{align*}
Appealing to any of the lower bounds on $D_M^n(\cdot)$ derived in the previous section results in a dimension free bound.

%%%%%%%%%%%%%%%%%%%%%%%%%%%%%%%%%%%%%%%%%%%%%%%%%%%%%%%%%%%%%%%%%%%
% VIII-B Quadratic bound
%%%%%%%%%%%%%%%%%%%%%%%%%%%%%%%%%%%%%%%%%%%%%%%%%%%%%%%%%%%%%%%%%%%

\subsection{A quantization bound for quadratic forms}
We now focus our attention on reversible Markov chains as a preparation for our improved bound on the isoperimetric constant in the next subsection.  
Let $W$ be a reversible Markov kernel on $[K]$ with invariant distribution $\mu$ and let $(X,Y)\sim\mu\times W$. 
Consider the quadratic form 
$$\E[f(X)f(Y)]=\langle f,Wf\rangle_\mu,$$
which is governed by the eigenvalues of the reversible kernel rather than the singular values of its joint law. Since $W$ is reversible, the Markov operator $f\mapsto\E[f(Y)\mid X{=}\cdot]$ is self-adjoint on $L^2(\mu)$, so by the spectral theorem it has real eigenvalues $1=\lambda_1(W)\geq\cdots\geq\lambda_K(W)\geq-1$ and a $\mu$-orthonormal eigenbasis $\{\varphi_i\}$ with $\varphi_1=1$. Expanding any zero-mean, unit-norm $f$ as $f=\sum_{i\geq2}\langle f,\varphi_i\rangle\varphi_i$, with $\sum_{i\geq2}\langle f,\varphi_i\rangle^2=1$, yields
\begin{align}
    \E[f(X)f(Y)]=\langle f,Wf\rangle_\mu=\sum_{i\geq2}\lambda_i(W)\,\langle f,\varphi_i\rangle^2.\label{eq:quadform}
\end{align} 
The entries of the DTM matrix of $\mu\times W$ are $\frac{P(x,y)}{\sqrt{\mu(x)\mu(y)}}=\sqrt{\frac{\mu(x)}{\mu(y)}}W(y|x)$, thus it is similar to $W$ as we can write it as $SWS^{-1}$ with $S=\text{diag}(\mu)^{1/2}$, and therefore the two share the same eigenvalues, $\lambda_i(\text{DTM})=\lambda_i(W)$. Moreover, the reversibility of $W$ implies the symmetry of the DTM matrix, as
\begin{align}
    \mu(x)W(y|x)=\mu(y)W(x|y) \Longleftrightarrow \sqrt{\frac{\mu(x)}{\mu(y)}}W(y|x)=\sqrt{\frac{\mu(y)}{\mu(x)}}W(x|y).\label{eq:dtm-symmetric}
\end{align}

When the DTM matrix is positive semidefinite, its singular value
decomposition coincides with its eigendecomposition: $\sigma_i =
\lambda_i \geq 0$, the singular functions on both sides equal the
eigenfunctions, $f_i = g_i = \varphi_i$, and the machinery of
Section~\ref{subsec:spectralDecomp} applies to $\mu \times W$ with eigenvalues in place of
singular values and a single family of functions.
Since the DTM matrix is symmetric with spectrum $\{\lambda_i(W)\}$, it is
positive semidefinite precisely when all eigenvalues of $W$ are
nonnegative.

Now, to bound \eqref{eq:quadform}, recall the projection-quantization duality of
Section~\ref{sec:upper}: for any closed subspace $H \subseteq L^2(\mu^{\otimes n})$
with $H \perp \mathbf{1}$, the maximal projection and optimal MMSE of
Definitions~\ref{def:maxproj} and~\ref{def:MMSEopt} satisfy $\phi^n_M(H)^2 + D^n_M(H) = 1$ by Theorem~\ref{thm:ProjQuantDuality}.
The subspace we apply this to is the one carrying the top nontrivial eigenvalue.

\begin{definition}[Second eigenspace]\label{def:H2}
Let $\mathsf{K}$ be a reversible kernel with eigenvalues $1=\lambda_1(\mathsf{K})\geq\lambda_2(\mathsf{K})\geq\cdots$ and $\mu$-orthonormal eigenbasis $\{\varphi_i\}$. The second eigenspace is the $\lambda_2(\mathsf{K})$-eigenspace,
\begin{align}
    H_2\triangleq\mathrm{span}\{\varphi_i:\lambda_i(\mathsf{K})=\lambda_2(\mathsf{K})\}.
\end{align}
\end{definition}

The following theorem is the quadratic-form counterpart of
Corollary~\ref{cor:rhoasmmse}: it bounds the agreement of an $M$-level
function with itself in terms of the eigenspectrum and the quantization
MMSE $D^n_M(H_2)$ onto the second eigenspace. 
Note that the theorem holds for a general reversible kernel, whose DTM matrix need not be positive semidefinite. The proof proceeds by passing to the lazy kernel $\frac{I+\mathsf{K}}{2}$, which is positive semidefinite, applying the machinery of Section~\ref{sec:upper} to it,
and converting the resulting bound back to the original kernel.

\begin{theorem}\label{thm:quad}
Let $\mathsf{K}$ be a reversible kernel on $[K^n]$ with invariant law $\mu^{\otimes n}$, second eigenspace $H_2$ as in Definition~\ref{def:H2}, and let
$\lambda_*(\mathsf{K})\triangleq\max\{\lambda_i(\mathsf{K}):\lambda_i(\mathsf{K})<\lambda_2(\mathsf{K})\}$
be the largest eigenvalue below $\lambda_2(\mathsf{K})$, with the convention $\lambda_*(\mathsf{K})=\lambda_2(\mathsf{K})$ if no eigenvalue lies below it. Then for any $f\in\mathcal{F}^n_M(\mu^{\otimes n})$, with $(X^n,Y^n)\sim\mu^{\otimes n}\times\mathsf{K}$,
\begin{align}
    \E[f(X^n)f(Y^n)]\;\leq\;\lambda_2(\mathsf{K})-\big(\lambda_2(\mathsf{K})-\lambda_*(\mathsf{K})\big)\,D^n_M(H_2).\label{eq:quadbound}
\end{align}
\end{theorem}
\begin{proof}
Let $\mathsf{K}' \triangleq \frac{I+\mathsf{K}}{2}$, and let $\tilde Y^n$
denote the output of $\mathsf{K}'$ on input $X^n$. For any unit-norm $f$,
\begin{align}
\E[f(X^n)f(\tilde Y^n)] = \langle f, \mathsf{K}' f\rangle
= \tfrac{1}{2}\bigl(1 + \E[f(X^n)f(Y^n)]\bigr), \label{eq:lazy-identity}
\end{align}
so it suffices to bound the left-hand side. $\mathsf{K}'$ is
reversible with invariant law $\mu^{\otimes n}$ and the same
eigenfunctions as $\mathsf{K}$, with eigenvalues $\lambda_i(\mathsf{K}')
= \frac{1+\lambda_i(\mathsf{K})}{2} \in [0,1]$; because $x \mapsto
\frac{1+x}{2}$ is increasing, the order of eigenvalues is preserved, so
$H_2(\mathsf{K}') = H_2(\mathsf{K})$ and $\lambda_*(\mathsf{K}') =
\frac{1+\lambda_*(\mathsf{K})}{2}$. As the eigenvalues of $\mathsf{K}'$
are nonnegative, its DTM matrix is positive semidefinite, so $\sigma_i =
\lambda_i(\mathsf{K}')$ in matching order and the matched subspaces of
the subsystem $I = \{i : \lambda_i(\mathsf{K}') =
\lambda_2(\mathsf{K}')\}$ coincide, $H_{P_X}(I) = H_{P_Y}(I) = H_2$. Since these
are the largest nontrivial singular values, $I$ is dominant.
Theorem~\ref{thm:MainResultSubsystems} applied with $f = g$, 
followed by Lemma~\ref{lem:simplify_func} with $(s,t) = (\lambda_2(\mathsf{K}'),\lambda_*(\mathsf{K}'))$ that is admissible as $0 \leq \lambda_*(\mathsf{K}') \leq \lambda_2(\mathsf{K}')$, 
and falling in the third case of \eqref{eq:SimplifyLemmaCases} since $x^* = y^* > a$ when $s > t$ and $a < 1$, with the boundary cases agreeing across branches, yields
\begin{align}
\E[f(X^n)f(\tilde Y^n)] \;\leq\; \lambda_2(\mathsf{K}') -
\bigl(\lambda_2(\mathsf{K}') - \lambda_*(\mathsf{K}')\bigr)
D^n_M(H_2).
\end{align}
Substituting into \eqref{eq:lazy-identity} and solving, using
$2\lambda_2(\mathsf{K}') - 1 = \lambda_2(\mathsf{K})$ and
$2\bigl(\lambda_2(\mathsf{K}') - \lambda_*(\mathsf{K}')\bigr) =
\lambda_2(\mathsf{K}) - \lambda_*(\mathsf{K})$, gives \eqref{eq:quadbound}.
\end{proof}

%%%%%%%%%%%%%%%%%%%%%%%%%%%%%%%%%%%%%%%%%%%%%%%%%%%%%%%%%%%%%%%%%%%
% VIII-C Isoperimetric bound
%%%%%%%%%%%%%%%%%%%%%%%%%%%%%%%%%%%%%%%%%%%%%%%%%%%%%%%%%%%%%%%%%%%
\subsection{An improved bound on the isoperimetric constant for cartesian product graphs and channels}
Let $W^{(n)}$ be the Cartesian product of $W$ over $[K^n]$~\cite{am85,houdre2004isoperimetric} defined in the following manner: A Markov kernel is chosen uniformly from $[n]$ and is incremented according to $W$, while keeping all other values fixed. Namely, for the input $x^n$, we have that the output $Y_J\sim W(\cdot | x_J)$ for $J\sim \mathrm{Unif}\{[n]\}$, and $Y_j = x_j $ for all $j\neq J$. It can be easily verified that the Cartesian product is also a reversible Markov kernel with a unique invariant distribution $\mu^n=\mu\otimes\cdots \otimes \mu$. Let $\{\lambda_k(L)\}_{k=1}^{K}$ be the eigenvalues of $L=I-W$ in increasing order, so $\lambda_1(L)=0$; equivalently, the eigenvalues of $W$ are $\lambda_i(W)=1-\lambda_i(L)$, in decreasing order.

For a reversible kernel $\mathsf{K}$ with invariant law $\pi$ and
$(X,Y) \sim \pi \times \mathsf{K}$, the isoperimetric (Cheeger) constant is
\begin{align}
h(\mathsf{K}) \;\triangleq\; \inf_{\substack{f : \mathcal{X} \to \{0,1\}}}
\frac{\Pr\bigl(f(X) \neq f(Y)\bigr)}{\min\{\Pr(f(X)=0), \Pr(f(X)=1)\}}.
\label{eq:hdef}
\end{align}

It is known, due to~\cite{am85,houdre2004isoperimetric,chung1998isoperimetric}, that
\begin{align}\label{eq:known_bound_on_h}
    h(W^{(n)}) \geq \frac{1}{n} \cdot\max \left\{\frac{h(W)}{2}, \lambda_2(L)\right\}.
\end{align}
We now establish Theorem~\ref{theorem:Isoperimetric}, which strengthens the spectral bound in~\eqref{eq:known_bound_on_h} in terms of the $\lambda_2$-eigenspace of $L$ and the \textit{second spectral gap} $\lambda_3(L)-\lambda_2(L)$.
\begin{theorem}[Isoperimetric inequality]\label{theorem:Isoperimetric}
Assume $\lambda_3(L)>\lambda_2(L)$. Then
\begin{align}
    h(W^{(n)})\geq\frac1n\Big(\lambda_2(L)+\min\{\lambda_2(L),\lambda_3(L)-\lambda_2(L)\}\cdot D^n_2(\varphi_2(X))\Big),\label{eq:iso}
\end{align}
where $\varphi_2$ is the $\mu$-normalized eigenfunction of $L$ associated with $\lambda_2(L)$, and $D^n_2$ is as in Definition~\ref{def:DnM} with $X^n\sim\mu^{\otimes n}$.
\end{theorem}

\begin{proof}
First, note that for a reversible kernel $\mathsf{K}$ with invariant law $\mu$, $(X,Y) \sim \mu \times \mathsf{K}$, and any $f : \mathcal{X} \to
\{0,1\}$ with $p \triangleq \Pr(f(X)=1) \in (0,1)$, we have:
\begin{align}
    \frac{\Pr\bigl(f(X) \neq f(Y)\bigr)}{\min\{p, 1-p\}}
    \geq 1 - \rho (f,f)
    \geq 1 - \max_{g\in\mathcal{F}_2(\mu)}\E[f(X)f(Y)],
    \label{eq:isoproof_A_1}
\end{align}
where the first inequality is by~\eqref{eq:AgreementProbWitsBound} with $q=p$ together with $\min\{p, 1-p\} \leq 2p(1-p)$,
and the second holds since, by linear invariance of the Pearson correlation, $\rho(f,f) = \E[\bar f(X)\bar f(Y)]$ for the standardized
$\bar f \triangleq (f-p)/\sqrt{p(1-p)} \in \mathcal{F}_2(\mu)$.
In particular, we have:
\begin{align}
    h(W^{(n)})\geq1-\max_{f\in\mathcal{F}^n_2(\mu^{\otimes n})}\E[f(X^n)f(Y^n)],\qquad (X^n,Y^n)\sim\mu^{\otimes n}\times W^{(n)},\label{eq:isoproof_A}
\end{align}
by minimizing \eqref{eq:isoproof_A_1} with $\mathsf{K} = W^{(n)}$.

To bound this maximum we invoke Theorem~\ref{thm:quad} with $\mathsf{K}=W^{(n)}$, whose eigenstructure is well known~\cite[Lemma 12.11]{PeresLevin2017markov}: its eigenfunctions are the products $\varphi_u(X^n)=\prod_{i=1}^n\varphi_{u_i}(X_i)$, indexed as in Section~\ref{subsec:productspace}, but with eigenvalues that are the coordinate averages $\lambda_u(W^{(n)})=\frac1n\sum_{i=1}^n\lambda_{u_i}(W)$ rather than products.
The positive gap makes $\lambda_2(W)$ simple, so the second eigenspace is spanned by its $n$ coordinate copies and by no other modes,
\begin{align}
    H_2=\mathrm{span}\{\varphi_{\mathds{1}+e_i}\}_{i=1}^n,\qquad \lambda_2(W^{(n)})=\tfrac1n\big(\lambda_2(W)+(n-1)\big).
\end{align}
This structure makes the distortion computable: 
since $H_2$ is spanned by the functions $\varphi_{\mathds{1}+e_i}(X^n)=\varphi_2(X_i)$, i.e., by i.i.d.\ coordinate copies of the scalar source $\varphi_2(X)$,
the distortion onto $H_2$ coincides with the scalar MMSE of quantizing a normalized combination of those copies, $D^n_2(H_2)=D^n_2(\varphi_2(X))$, which is the quantity bounded in Sections~\ref{sec:ratedist} and~\ref{sec:anticont}.\footnote{A situation where $D^n_2(H_2)\neq D^n_2(\varphi_2(X))$ requires $\lambda_2(W)$ to have multiplicity greater than one, which happens iff $\lambda_2(W)=\lambda_3(W)$. This is the zero-gap boundary $\lambda_2(L)=\lambda_3(L)$ excluded from the Theorem's premise: In such case, the eigenspace $H_2$ is spanned by coordinate copies of several distinct eigenfunctions rather than of $\varphi_2$ alone, so $D^n_2(H_2)$ is no longer the scalar MMSE of a single source and the reduction fails. One way or the other, failing to either capture all of $H_2$ or to bound its distortion away from zero results in the coefficient $\min\{\lambda_2(L),\lambda_3(L)-\lambda_2(L)\}\cdot D^n_2$ in~\eqref{eq:iso} vanishing, and the bound collapses to the prior art~\eqref{eq:known_bound_on_h}.} The largest eigenvalue below $\lambda_2(W^{(n)})$ is attained at $u=\mathds{1} +2e_i$ or $u=\mathds{1}+e_i+e_j,\ i\neq j$, giving
\begin{align}
    \lambda_*(W^{(n)})=\tfrac1n\big(\max\{\lambda_3(W),\,2\lambda_2(W)-1\}+(n-1)\big).
\end{align}
Substituting these into~\eqref{eq:quadbound} with $M=2$,
\begin{align}
    \max_{f}\E[f(X^n)f(Y^n)]\leq\tfrac1n\big(\lambda_2(W)+(n-1)\big)-\tfrac1n\big(\lambda_2(W)-\max\{\lambda_3(W),2\lambda_2(W)-1\}\big)D^n_2(\varphi_2(X)).
\end{align}
Plugging this into~\eqref{eq:isoproof_A} and rearranging, we get
\begin{align}
    h(W^{(n)})\geq\tfrac1n(1-\lambda_2(W))+\tfrac1n\big(\lambda_2(W)-\max\{\lambda_3(W),2\lambda_2(W)-1\}\big)D^n_2(\varphi_2(X)).\label{eq:iso-in-W}
\end{align}
It remains to convert~\eqref{eq:iso-in-W} to the $L$-spectrum via $\lambda_i(W)=1-\lambda_i(L)$. The leading term is $1-\lambda_2(W)=\lambda_2(L)$. For the coefficient, $\max\{\lambda_3(W),2\lambda_2(W)-1\}=1-\min\{\lambda_3(L),2\lambda_2(L)\}$, so
\begin{align}
    \lambda_2(W)-\max\{\lambda_3(W),2\lambda_2(W)-1\}=\min\{\lambda_3(L),2\lambda_2(L)\}-\lambda_2(L)=\min\{\lambda_3(L)-\lambda_2(L),\lambda_2(L)\},
\end{align}
which is exactly inequality~\eqref{eq:iso}.
\end{proof}

%%%%%%%%%%%%%%%%%%%%%%%%%%%%%%%%%%%%%%%%%%%%%%%%%%%%%%%%%%%%%%%%%%%
% ISO Example 7
%%%%%%%%%%%%%%%%%%%%%%%%%%%%%%%%%%%%%%%%%%%%%%%%%%%%%%%%%%%%%%%%%%%

\begin{example}
Let the kernel $W$ be given by
\begin{align}
  W(y|x) = \frac{(x+y \bmod 3)+1}{6}, \qquad x,y\in\{0,1,2\}.
\end{align}
Since $W(y\mid x)=W(x\mid y)$, the kernel is symmetric and hence doubly stochastic, so its invariant distribution $\mu$ is uniform and detailed balance
holds, making $W$ reversible.

We first compute $h(W)$. Any binary $f:\{0,1,2\}\to\{0,1\}$ conflates exactly two
states, so under the uniform $\mu$ we have
$\min\{\Pr(f(X)=0),\Pr(f(X)=1)\}=\tfrac13$. As $W(y\mid x)$ is largest when
$x+y=2$, the disagreement probability is minimized by the assignment
$\{0,2\}\mapsto 1,\ 1\mapsto 0$, giving
\begin{align}
    h(W)=3\Pr(f(X)\neq f(Y))=W(1|0)+W(\{0,2\}|1)+W(1|2)=1.
\end{align}
The eigendecomposition of $L$ gives $\lambda_2(L)\approx 0.7113$ and
$\lambda_3(L)\approx 1.289$, so $\min\{\lambda_2(L),\lambda_3(L)-\lambda_2(L)\}
=\lambda_3(L)-\lambda_2(L)\approx 0.577$, with associated $\mu$-normalized
eigenfunction
\begin{align}
    f_2(x) =
  \begin{cases}
    \frac{-1+\sqrt3}{2}, & x=0,\\[3pt]
    \frac{-1-\sqrt3}{2}, & x=1,\\[3pt]
    1, & x=2.
  \end{cases}
\end{align}
The classical bound~\eqref{eq:known_bound_on_h} gives $h(W)\ge\lambda_2(L)\approx 0.7113$. Our scalar bound, Theorem \ref{theorem:Isoperimetric} at $n=1$, reads
\begin{align}
  h(W) \ge \lambda_2(L) + \min\{\lambda_2(L),\lambda_3(L)-\lambda_2(L)\}\cdot D_2(f_2(X)).
\end{align}

To compute $D_2(f_2(X))$, recall that the MSE-optimal two-level quantizer of a scalar source is a threshold function. 
Since $f_2(1) < f_2(0) < f_2(2)$, the only two thresholds isolate either the smallest value, $\{1\}\mid\{0,2\}$, or the largest, $\{0,1\}\mid\{2\}$.
The former is MSE-optimal, with conditional-mean representatives $\bigl(-\tfrac{1+\sqrt3}{2},\,\tfrac{1+\sqrt3}{4}\bigr)$ and
distortion
\begin{align}
  D_2(f_2(X)) = \tfrac{2}{3}\Bigl(\tfrac{3-\sqrt{3}}{4}\Bigr)^2 \approx 0.067,
\end{align}
against a distortion of $0.5$ for the latter. 
The scalar bound therefore gives
$h(W)\ge \lambda_2(L) + 0.577\cdot 0.067 = \lambda_2(L) + 0.0387$. Note that $\lambda_2(L)$ already exceeds the bound provided by~\cite{erbar2018poincare}, namely $\frac{1}{3}\sqrt{\lambda_2(L)\cdot\min_{x,y}W(y|x)}$, and so clearly does our bound.

For $n>1$ we appeal to Corollary~\ref{cor:unif} with $U=f_2(X)$. Since $H(f_2(X))=\log 3$ and
$d_{\min}\approx 0.634$,
\begin{align}
  D_2^n(f_2(X)) \ge \frac{d_{\min}^2}{12}
    \Bigl(\frac{12}{2\pi e}\cdot\frac{2^{2H(f_2(X))}}{2^2}-1\Bigr) \approx 0.019,
\end{align}
which yields
\begin{align}
  h(W^{(n)}) \ge \frac1n\Bigl(\lambda_2(L) + \min\{\lambda_2(L),\lambda_3(L)-\lambda_2(L)\}
    \cdot D_2^n(f_2(X))\Bigr) = \frac1n\bigl(\lambda_2(L) + 0.011\bigr),
\end{align}
strictly improving, for every $n$, the classical bound~\eqref{eq:known_bound_on_h},
which for this chain equals $\tfrac{1}{n}\lambda_2(L)$ (as $\lambda_2(L) > h(W)/2$).
\end{example}

\section{Acknowledgments}  
This work was supported by the ISF under Grants 1791/17, 1495/18, and 1766/22. The work of DD was further supported by the Yitzhak and Chaya Weinstein Research Institute for Signal Processing.

\bibliographystyle{IEEEtran}
{\footnotesize
\bibliography{references}}

\begin{appendices}
\section{Proof of Lemma~\ref{lem:simplify_func}}
We first find the extremum of $g(x,y)=x y  s  + \sqrt{(1-x^2)(1-y^2)} t$. We have
\begin{align}
    \frac{\partial g}{\partial x}&=y s +\frac{\sqrt{1-y^2}}{2\sqrt{1-x^2}}\cdot (-2xt),\\\frac{\partial g}{\partial y}&=x s +\frac{\sqrt{1-x^2}}{2\sqrt{1-y^2}}\cdot (-2yt).
\end{align}
The only extremum point of $g$ is $(0,0)$, which is clearly a minimum. Thus, the maximum is attained on the boundaries, either on $x=0$, $x=a$, $y=0$ or $y=b$.
We have $\max_{0\leq y\leq b}g(0,y)=\max_{0\leq y\leq b}\sqrt{1-y^2}t=t$ and, similarly, $\max_{0\leq x\leq a}g(x,0)=\max_{0\leq x\leq a}\sqrt{1-x^2}t=t$. These values never exceed the maximum over the remaining boundaries, since, letting $c=\min\{a,b\}$, we can set $x=y=c$ to get
\begin{align}
   g(c,c)=c^2  s +(1-c^2)t\geq c^2 t+(1-c^2)t=t,
\end{align}
as $ s \geq t$. Thus, it is sufficient to look for the maximum of the simple functions $g(y)=g(a,y)$ and $g(x)=g(x,b)$:
\begin{align}
   \frac{dg(x)}{dx}= b s -t\sqrt{1-b^2}\frac{x}{\sqrt{1-x^2}}.
\end{align}
Equating to zero, we have 
\begin{align}
    \frac{x^2}{1-x^2}=\frac{(b s )^2}{t^2(1-b^2)},
\end{align}
implying that 
\begin{align}
   x^*=\frac{b s }{\sqrt{b^2 s ^2+(1-b^2)t^2}} .
\end{align}
Similarly, optimizing over $g(y)$, we obtain
\begin{align}
  y^*=\frac{a s }{\sqrt{a^2 s ^2+(1-a^2)t^2}} . 
\end{align}
The maximum is thus either 
\begin{align}
  g(a,y^*)&= a y^*  s  + \sqrt{(1-a^2)(1-(y^*)^2)} t\\&=\frac{1}{\sqrt{a^2 s ^2+(1-a^2)t^2}}\cdot (a^2 s ^2+(1-a^2) t^2\\&=\sqrt{a^2 s ^2+(1-a^2)t^2},
\end{align}
or, similarly, $g(x^*,b)=\sqrt{b^2 s ^2+(1-b^2)t^2}$. Finally, note that $b > a$ implies $g(x^*,b) \geq g(a,y^*)$ and that $b\leq a$ implies $g(x^*,b)\leq g(a,y^*)$. Thus if $\{b> a,x^*\leq a,y^*\leq b\}$ or $\{x^*\leq a,y^*> b\}$ we have $d(a,b)=\sqrt{b^2 s ^2+(1-b^2)t^2}$. Otherwise, if $\{b\leq a,x^*\leq a,y^*\leq b\}$ or $\{x^*> a,y^*\leq b\}$ we have $d(a,b)=\sqrt{a^2 s ^2+(1-a^2)t^2}$. These conditions do not hold only if the maximum is attained at some point outside the rectangle $\{0\leq x\leq a,0\leq y\leq b\}$. In this case, due to the monotonicity of $g(a,y)$ (or $g(x,b)$), the maximum is attained at the edge point $(x,y)=(a,b)$, that is, $d(a,b)=ab s + \sqrt{(1-a^2)(1-b^2)} t$ .
%%%%%%%%%%%%%%%%%%%%%%%%%%%%%%%%%%%%%
%%%%%%%%%%%%%%%%%%%%%%%%%%%%%%%%%%%%%
\section{Proof of Lemma~\ref{lemma:DIMINDEP_A}}\label{app:proof_lemma5}
Note that $E(\alpha U, t) = E(U,t/\alpha)$, and that if $Z = \alpha U$, then the symmetrized variable $Z^*$ corresponding to $Z$ has the same distribution as $\alpha U^*$ (or $-\alpha U^*$). 
Now, let $a \in \mathbb{R}^n$ with $\|a\|_2=1$. By Lemma~\ref{lem:Petrov}
\begin{align*}
    &Q\left(\sum a_i U_i ,t \right) \leq  \frac{Ct}{\sqrt{\sum_{i=1}^n t_i^2 E(\left(a_i U_i\right)^*;t_i)}} \\
    &= \frac{Ct}{\sqrt{\sum_{i=1}^n t_i^2 E(|a_i| \left(U_i\right)^*;t_i)}} = \frac{Ct}{\sqrt{\sum_{i=1}^n t_i^2 E(U^*_i;\frac{t_i}{|a_i|})}},
\end{align*}
Without losing generality, assume that $|a_1| \geq |a_i|, \  1 \leq i \leq n$.
Pick $t_i = t |a_i|/c$ for some $c > |a_1|$, then
\begin{align}
\nonumber    Q\left(\sum a_i U_i ,t \right) &\leq \frac{Ct}{\sqrt{\sum_{i=1}^n (t |a_i|/c)^2 E(U^*;\frac{t}{c})}}\\
    & = C\cdot c / \sqrt{E(U^*;t/c)}.\label{eq:PetrovStillApplicable}
\end{align}
Next, let $t_\tau = c \cdot \tau$ for $c = C^{-1}Q(U,\tau) \sqrt{E(U^*;\tau)}$ and some $\tau > 0$. 
Note that this choice does not necessarily satisfy $c > |a_1|$, which is a necessary condition in Lemma \ref{lem:Petrov}. 
But if $c > |a_1|$ nevertheless, then Lemma \ref{lem:Petrov} and \eqref{eq:PetrovStillApplicable} give
\begin{align}
    Q\left(\sum a_i U_i ,t_\tau \right)  \leq Q(U,\tau).
 \label{eq:TwoBoundsSameResult}
\end{align}
%%%%%%%%%%%%%%%%%%%%%%%%%%%%%%%%%%%
%         1-Column-Bound.         %
%%%%%%%%%%%%%%%%%%%%%%%%%%%%%%%%%%%
If however $c \leq |a_1|$, a different argument applies. Consider the concentration of the random variable $|a_1| U_1$. A consequence of inequality \eqref{eq:lemmaMinimumConcentration} is that $Q\left(\sum a_i U_i ,t \right) \leq Q(|a_1| U,t)$. Furthermore, by our assumption that $c \leq |a_1|$ we have
\begin{align*}
    Q\left(\sum a_i U_i ,t \right) &\leq Q(|a_1| U,t)  =  Q\left(U,\frac{t}{|a_1|}\right) \leq Q\left(U,\frac{t}{c}\right),
\end{align*}
and thus, under the same choice of $t_\tau = c \cdot \tau$, we get~\eqref{eq:TwoBoundsSameResult} again.
We therefore conclude that~\eqref{eq:TwoBoundsSameResult} holds whether or not $c \geq |a_1|$, where $t_\tau = Q(U,\tau) \sqrt{E(U^*;\tau)}C^{-1}\tau.$ This holds for any $a\in \mathbb{R}^n$ with $\|a\|_2=1$, hence $A_n\left(U,t_\tau \right) \leq Q(U,\tau)$. 

\end{appendices}
\end{document}